\documentclass[journal,comsoc]{IEEEtran}
\IEEEoverridecommandlockouts
\usepackage[normalem]{ulem}
\usepackage{cite}
\usepackage{graphicx}
\usepackage[font=footnotesize,skip=0pt]{caption}
\usepackage{textcomp}
\usepackage[utf8]{inputenc}
\usepackage{subfigure}
\usepackage{enumitem}
\usepackage[T1]{fontenc}
\usepackage{graphicx,booktabs,multirow,bm}
\newcommand*{\TT}{{\mathrm{T}}}

\newcommand*{\s}{{\mathrm{s}}}
\usepackage{algorithm}
\usepackage{algorithmic}
\usepackage{array}
\usepackage{tabularx}
\usepackage{amsmath,amssymb,amsfonts,amsthm}
\usepackage{bbm}
\theoremstyle{plain}
\newtheorem{definition}{\bf{Definition}}
\newtheorem{assumption}{\bf{Assumption}}
\newtheorem{theorem}{\bf{Theorem}}

\newtheorem{lemma}{\bf{Lemma}}
\newtheorem{remark}{\bf{Remark}}
\usepackage{xcolor}
\usepackage{bbding}
\usepackage{hyperref}
\usepackage{url}
\usepackage{cite}
\def\BibTeX{{\rm B\kern-.05em{\sc i\kern-.025em b}\kern-.08em
    T\kern-.1667em\lower.7ex\hbox{E}\kern-.125emX}}
\allowdisplaybreaks[4]

\begin{document}

\title{Stochastic Coded Federated Learning: Theoretical Analysis and Incentive Mechanism Design}

\author{Yuchang~Sun,~\IEEEmembership{Student Member,~IEEE},
Jiawei~Shao,~\IEEEmembership{Student Member,~IEEE}, 
Yuyi~Mao,~\IEEEmembership{Member,~IEEE},
Songze~Li,~\IEEEmembership{Member,~IEEE}, 
and Jun~Zhang,~\IEEEmembership{Fellow,~IEEE}

\thanks{
      	Manuscript received 8 November 2022; revised 17 May and 6 September 2023; accepted 8 November 2023. This paper was presented in part at the 2022 IEEE International Symposium on Information Theory \cite{scfl} [DOI: 10.1109/ISIT50566.2022.9834445]. The editor coordinating the review of this article and approving it for publication was Z. Cai. \emph{(Corresponding authors: Jun Zhang and Yuyi Mao)}
      	
      	Y. Sun, J. Shao, and J. Zhang are with the Department of Electronic and Computer Engineering, the Hong Kong University of Science and Technology, Hong Kong, China (E-mail: \{yuchang.sun, jiawei.shao\}@connect.ust.hk, eejzhang@ust.hk). 
      	Y. Mao is with the Department of Electrical and Electronic Engineering, the Hong Kong Polytechnic University, Hong Kong, China (E-mail: yuyi-eie.mao@polyu.edu.hk).
      	S. Li is with School of Cyber Science and Engineering, Southeast University, Nanjing, China (E-mail: songzeli@seu.edu.cn).

       The work of S. Li is in part supported by the National Nature Science Foundation of China (NSFC) Grant 62106057.
}}

\maketitle

\begin{abstract}
Federated learning (FL) has achieved great success as a privacy-preserving distributed training paradigm, where many edge devices collaboratively train a machine learning model by sharing the model updates instead of the raw data with a server. However, the heterogeneous computational and communication resources of edge devices give rise to stragglers that significantly decelerate the training process. To mitigate this issue, we propose a novel FL framework named stochastic coded federated learning (SCFL) that leverages coded computing techniques. In SCFL, before the training process starts, each edge device uploads a privacy-preserving coded dataset to the server, which is generated by adding Gaussian noise to the projected local dataset. During training, the server computes gradients on the global coded dataset to compensate for the missing model updates of the straggling devices. We design a gradient aggregation scheme to ensure that the aggregated model update is an unbiased estimate of the desired global update. Moreover, this aggregation scheme enables periodical model averaging to improve the training efficiency. We characterize the tradeoff between the convergence performance and privacy guarantee of SCFL. In particular, a more noisy coded dataset provides stronger privacy protection for edge devices but results in learning performance degradation. We further develop a contract-based incentive mechanism to coordinate such a conflict. The simulation results show that SCFL learns a better model within the given time and achieves a better privacy-performance tradeoff than the baseline methods. In addition, the proposed incentive mechanism grants better training performance than the conventional Stackelberg game approach.
\end{abstract}


\section{Introduction}

The rapid growth of artificial intelligence (AI) technologies is boosting the development of various intelligent applications such as digital healthcare, smart transportation, and augmented/virtual reality (AR/VR).
These applications generate a large volume of data at the wireless network edge, which contain valuable information for training high-quality AI models to improve user experience.
Whilst the traditional solution is to directly upload the data to a cloud server for centralized training, it is prohibitive in many emerging use scenarios where the data may contain privacy-sensitive information, e.g., geographical locations and user preferences \cite{meneghello2019iot,li2021security}.
Federated learning (FL) \cite{fedavg}, which iteratively aggregates local model updates computed by the edge devices (e.g., smartphones and Internet of Things (IoT) devices) at a server, was proposed as a promising framework for privacy-preserving distributed model training.

The training process of FL is divided into multiple communication rounds.
In each communication round, the edge devices perform local model training with their own data and the server aggregates their model updates to generate a new global model for the next round.
Due to the heterogeneous computational and communication resources at different edge devices, e.g., computing speeds and network link rates, it may take a much longer time for edge devices with fewer resources to finish the local training \cite{nishio2019client}.
Since the local data among edge devices are typically non-independent and identically distributed (non-IID), simply ignoring the updates from stragglers may lead to biased model aggregations and greatly degrade the training performance.
To mitigate the straggler effect, various methods such as device scheduling \cite{cs1,cs3} and asynchronous model aggregation \cite{AFLsurvey} have been developed.
While the training efficiency can be improved to some extent, these solutions fail to exploit the valuable computational resources at the server, which can be utilized to further accelerate the training speed.

\emph{Coded federated learning} (CFL) \cite{CFL,CFL_journal}, which exploits the server's computational resources by constructing coded datasets at the edge devices \cite{li2020coded,codingnew2}, has recently been proposed to alleviate the straggler issue in FL.
Specifically, with the local coded datasets shared by the edge devices, the server is able to compute gradient updates to complement the missing ones of the stragglers.
Unfortunately, although the server cannot decode the raw data samples, the coded datasets still carry much information that causes privacy concerns on CFL.
In this paper, we propose a \textit{stochastic coded federated learning} (SCFL) framework with both convergence and privacy guarantee, where the local coded datasets are protected by both random projection and additive noise.
Also, we develop an incentive mechanism for SCFL to motivate edge devices to share less noisy coded data.

\subsection{Related Works}

To resolve the straggler issue in FL, a simple approach is to ignore the straggling edge devices in each communication round.
In particular, the server can be configured to wait for a given duration in each communication round and aggregate the gradient updates received before the round deadline \cite{nishio2019client}.
Nevertheless, the presence of non-IID data results in biased aggregations that deviate from the global training objective.
Alternatively, asynchronous FL \cite{AFLsurvey} can be adopted where the model aggregation is triggered once sufficient model updates are available at the server.
Unfortunately, asynchronous FL may need more communication rounds to converge because the stale updates can be poisonous for the global model.
Also, some works \cite{cs1,cs3} designed device scheduling and resource allocation strategies to collect more informative updates in each communication round.
Besides, assigning models with heterogeneous sizes to edge devices according to their computation capabilities is also an effective approach to mitigate the straggler issue in FL \cite{diao2020heterofl}.
In addition, a barrier control method was proposed in \cite{zhao2021federated} to deal with the device heterogeneity in FL, where the server decides whether the collected local updates from edge devices should be passed to the aggregator or blocked to wait for further updates.
While existing approaches can improve the training efficiency, it is still likely that some edge devices cannot upload their gradient updates in each communication round, which leads to missing information for model aggregation at the server.
To combat the stragglers in conventional distributed computing systems, coded computing techniques \cite{li2020coded,codingnew2} introduce redundancy via coding theory when assigning tasks to distributed computing nodes.
Specifically, the server encodes the computation tasks with additional redundancy and distributes the tasks to the computing nodes such that it can recover the desired result from partial nodes.
Inspired by coded computing, a coded dataset was adopted to combat the straggler issue of FL \cite{CFL,CFL_journal,shao2022dres} in CFL, where local data samples of edge devices are encoded and uploaded to the server such that the server can compute gradients on the coded data to compensate for the missing updates from stragglers.
Compared with coded computing, the coded dataset in CFL are transformed from the local data at edge devices, which cannot be accessed by the server.
Therefore, a dedicated design is required to ensure unbiased gradient estimation at the server.
The first CFL schemes, including CFL-FB \cite{CFL} and CodedFedL \cite{CFL_journal}, construct the coded datasets by applying random linear projection on the local data.
In particular, the CFL-FB \cite{CFL} scheme performs full-batch (FB) gradient descent at both the server and edge devices, while a mini-batch of the local data are sampled to reduce the computation latency in each communication round \cite{CFL_journal}.
To ensure convergence, the server has to collect the gradients from edge devices after every local training step, which incurs significant communication overhead.
More recently, to enhance privacy protection, a differentially private coded federated learning (DP-CFL) \cite{anand2021differentially} scheme, which adds Gaussian noise to the coded data, was proposed.
However, the added noise leads to biased gradient estimates at the server and thus deteriorates the training performance.
The efficiency of the existing CFL schemes also suffers from frequent communication \cite{CFL,CFL_journal} or inadequate cooperation between edge devices and the server \cite{anand2021differentially}.
In addition, the interplay between privacy leakage in coded data sharing and training performance was not well studied.
The main characteristics of the existing CFL frameworks are summarized in Table \ref{table:related}.

\begin{table*}[!t]
\caption{Comparisons of different CFL frameworks.}
\label{table:related}
\centering
\resizebox{\textwidth}{!}{
\begin{tabular}{cccccccc}
\hline
\textbf{Framework} & {\begin{tabular}[c]{@{}c@{}}\textbf{Privacy Protection}\\ \textbf{of Coded Data}\end{tabular}} & {\begin{tabular}[c]{@{}c@{}}\textbf{Mini-batch} \\ \textbf{SGD}\end{tabular}} & {\begin{tabular}[c]{@{}c@{}} \textbf{Periodical} \\\textbf{Aggregation}\end{tabular}} &
 {\begin{tabular}[c]{@{}c@{}}\textbf{Convergence} \\ \textbf{Analysis}\end{tabular}} &
  {\begin{tabular}[c]{@{}c@{}}\textbf{Privacy}\\ \textbf{Analysis}\end{tabular}} & {\begin{tabular}[c]{@{}c@{}}\textbf{Tradeoff}\\ \textbf{Characterization}\end{tabular}} & {\begin{tabular}[c]{@{}c@{}}\textbf{Incentive}\\ \textbf{Mechanism}\end{tabular}} \\ \hline
CFL-FB \cite{CFL}    & Low  & \XSolidBrush   & \XSolidBrush    & \XSolidBrush     & \XSolidBrush   & \XSolidBrush  & \XSolidBrush \\
CodedFedL \cite{CFL_journal} & Low  & \Checkmark   & \XSolidBrush   & \Checkmark   & \Checkmark & \XSolidBrush  & \XSolidBrush \\
DP-CFL \cite{anand2021differentially}    & High &  \XSolidBrush  & N/A   & \XSolidBrush    &  \Checkmark     & \XSolidBrush  &  \XSolidBrush \\
SCFL (Ours)      & High &  \Checkmark   & \Checkmark   & \Checkmark   & \Checkmark    & \Checkmark  & \Checkmark \\ \hline
\end{tabular}}
\end{table*}

The noise levels of the coded datasets are key parameters in CFL. On one hand, to enhance privacy protection, the edge devices prefer adding stronger noise to the coded data.
On the other hand, from the perspective of the server, it is desirable to learn a better model by paying rewards to the edge devices for sharing less noisy coded datasets.
It is therefore necessary to design an incentive mechanism for determining the noise levels of coded datasets in CFL \cite{sun2022profit,ying2020double}.
To the best of our knowledge, the only incentive mechanism for CFL was proposed in \cite{ng2021hierarchical} for data allocation and device selection.
To determine the noise levels for gradient sharing in conventional FL, a Stackelberg game approach was developed in \cite{hu2020trading}, where the reward declared by the server is allocated to the edge devices according to their privacy costs.
Although this approach can be adapted for CFL, it excludes those edge devices with larger privacy sensitivity from uploading the coded dataset or compels them to add too strong noise to the coded datasets, which is detrimental to the learning performance.
To avoid this issue, we resort to the contract theory and develop a contract-based incentive mechanism \cite{bolton2004contract,kang2019incentive,limcontract}, where a contract item is designed for each edge device so that the server reduces the noise in the coded datasets with a minimal cost.

\subsection{Contributions}

In FL, it may take a much longer time for edge devices with fewer resources to finish the local training in each communication round, which leads to the straggling effect that significantly degrades the training efficiency.
Although ignoring the stragglers is a simple remedy, biased model aggregations may be incurred, which shall compromise the training performance.
In this work, we aim to tackle the straggler issue in FL via an innovative framework of SCFL, where a coded dataset is constructed at the server to generate unbiased gradient estimates that compensate for the missing gradients from stragglers.
Our main contributions are summarized as follows:
\begin{itemize}
    \item We propose an SCFL framework for federated linear regression. 
    Before training starts, each edge device generates a noisy local coded dataset via random projection, which is uploaded to the server.
    During training, the edge devices perform local training for multiple steps and upload the accumulated updates to the server for aggregation.
    Meanwhile, the server computes stochastic gradients based on the coded data. 
    Notably, the added noise reduces the privacy leakage of the local data compared with existing CFL frameworks \cite{CFL,CFL_journal}.
    Compared with our previous work \cite{scfl}, multiple steps of local training at both the server and edge devices in each communication round are allowed to improve the training efficiency.
    \item To ensure unbiased gradient estimation, we develop a novel gradient aggregation scheme at the server.
    We prove the convergence of SCFL and characterize the privacy guarantee of the coded datasets via the notion of mutual information differential privacy (MI-DP).
    Besides, by analyzing the effect of noise levels of the coded datasets, we theoretically demonstrate a tradeoff between privacy and performance of SCFL.
    \item 
    We further design a contract-based incentive mechanism to motivate edge devices to share coded datasets with less noise.
    Specifically, it derives a set of contract items with each specifying the privacy budget and the earned reward for each edge device.
    Accordingly, the mutually satisfactory noise levels of local coded datasets are determined with the minimal rewards paid by the server.
    \item We evaluate the proposed SCFL framework on the MNIST \cite{mnist} and CIFAR-10 \cite{cifar10} datasets.
    The simulation results demonstrate the benefits of SCFL in securing a higher test accuracy within the given training time.
    Compared with other CFL schemes, SCFL also achieves a consistently better privacy-performance tradeoff. 
    Besides, we also investigate the impacts of various system parameters, including the number of local steps and the ratio of stragglers, on the learning performance.
    In addition, the proposed contract-based incentive mechanism for SCFL yields a better learned model compared with the Stackelberg game approach for a given amount of reward paid by the server.
\end{itemize}

\subsection{Organization}

The rest of this paper is organized as follows.
In Section \ref{sec:system}, we describe the system model and the conventional FL algorithm for linear regression.
Section \ref{sec:algorithm} introduces the proposed SCFL framework.
We analyze the convergence and privacy guarantee of SCFL in Section \ref{sec:theory}.
An incentive mechanism for determining the noise levels of the coded datasets is developed in Section \ref{sec:incentive}.
Section \ref{sec:experiment} presents the experiment results and Section \ref{sec:conclusion} concludes this paper.

\subsection{Notations}
We use boldface upper-case letter (e.g., $\mathbf{X}$) and boldface lower-case letter (e.g., $\mathbf{x}$) to represent matrix and vector, respectively. 
For matrix $\mathbf{X}\in\mathbb{R}^{m \times n}$, $\mathbf{X}_{i,j}$ is the entry in the $i$-th row and $j$-th column, and $\|\mathbf{X}\|_{\mathrm{F}} \triangleq \sqrt{ \sum_{i=1}^{m}\sum_{j=1}^{n} \left|\mathbf{X}_{i,j}\right|^{2} }$ denotes its Frobenius norm.
Besides, we denote the $i$-th diagonal entry of a diagonal matrix $\mathbf{Z}$ as $z_i$.
The set of integers $\{1,2,\dots,N\}$ is denoted as $\left[ N \right]$.
In addition, $\mathbbm{1}\{\cdot\}$ is the indicator function, i.e., $\mathbbm{1}\{A\}=1$ if event $A$ happens and $\mathbbm{1}\{A\}=0$ otherwise.
Table \ref{tab:notation} summarizes the main notations in this paper.

\begin{table*}[t]
\caption{Main Notations}
\label{tab:notation}
\centering
\begin{tabular}{|c|c|c|c|}
\hline
\textbf{Notation} & \textbf{Meaning} & \textbf{Notation} & \textbf{Meaning}        \\ \hline
$(\mathbf{X}^{(i)}, \mathbf{Y}^{(i)})$      &  Local dataset of edge device $i$  & $(\mathbf{\hat{X}}_{k,u}^{(i)}, \mathbf{\hat{Y}}_{k,u}^{(i)})$    & Sampled data in step $u$ of round $k$ on edge device $i$ \\
$(\mathbf{\tilde{X}}^{(i)}, \mathbf{\tilde{Y}}^{(i)})$    & Coded dataset of edge device $i$  & $(\mathbf{\hat{X}}_{k,u}^{\s}, \mathbf{\hat{Y}}_{k,u}^{\s})$    & Sampled data in step $u$ of round $k$ on the server \\
$(\mathbf{X}, \mathbf{Y})$      & Global dataset    & $\mathbf{G}$        & Random projection matrix \\ 
$(\mathbf{\tilde{X}}, \mathbf{\tilde{Y}})$    & Global coded dataset & $\mathbf{N}$     & Gaussian noise matrix    \\ 
\hline
\end{tabular}
\end{table*}

\section{Federated Learning for Linear Regression}\label{sec:system}

\subsection{System Model}

We consider an edge AI system consisting of a server and $N$ edge devices. 
Each edge device $i\in [N]$ has a local dataset denoted as $\left( \mathbf{X}^{(i)}, \mathbf{Y}^{(i)} \right)$, where $\mathbf{X}^{(i)}\in \mathbb{R}^{l_i\times d}$ concatenates the $d$-dimensional features of $l_i$ data samples, and $\mathbf{Y}^{(i)} \in \mathbb{R}^{l_i\times o}$ represents the corresponding $o$-dimensional labels.
Accordingly, the global dataset in the FL system is given as the union of all the local datasets, i.e., $\mathbf{X} = \left[ (\mathbf{X}^{(1)})^{\TT},(\mathbf{X}^{(2)})^{\TT}, \dots,(\mathbf{X}^{(N)})^{\TT} \right]^{\TT} \in \mathbb{R}^{m \times d}$ and $\mathbf{Y} = \left[ (\mathbf{Y}^{(1)})^{\TT},(\mathbf{Y}^{(2)})^{\TT}, \dots,(\mathbf{Y}^{(N)})^{\TT} \right]^{\TT} \in \mathbb{R}^{m \times o}$, where $m = \sum_{i=1}^{N} l_{i}$ is the total number of data samples in the system.
We use matrix $\mathbf{Z}^{(i)} \triangleq \left[ z_{j^{\prime},j}^{(i)} \right] \in \mathbb{R}^{l_i \times m}$ to indicate the data availability on edge device $i$, where $z_{j^\prime,j}^{(i)}=1$ if the $j$-th data sample of the global dataset corresponds to the $j^\prime$-th data sample at edge device $i$ and $z_{j^\prime,j}^{(i)}=0$ otherwise. 
Thus, the local datasets can be expressed as $\left( \mathbf{X}^{(i)}, \mathbf{Y}^{(i)} \right) = \left( \mathbf{Z}^{(i)}\mathbf{X}, \mathbf{Z}^{(i)}\mathbf{Y} \right), i\in\left[N\right]$.

The server aims to fit a model $\mathbf{W} \in \mathbb{R}^{d\times o}$ to the distributed data on the edge devices by solving the following linear regression problem \cite{regression}:
\setlength\abovedisplayskip{0.1cm}
\setlength\belowdisplayskip{0.1cm}
\begin{equation}
    \min_{\mathbf{W} \in \mathbb{R}^{d\times o}} f(\mathbf{W}) \triangleq \frac{1}{2} \|\mathbf{X}\mathbf{W} - \mathbf{Y} \|_{\mathrm{F}}^{2}.
\label{eq:loss}
\end{equation}
We assume the server pays the edge devices for model training, and uses the learned model to earn profits, which is relevant to the business models of multiple industries.
For example, a software developer working for a news recommendation company can deploy a server to train a machine learning model based on the subscribers' behaviors \cite{news}, which enables high accuracy and thus increases the user viscosity.
Notably, if the global dataset is available at the server, the optimal linear regression model can be obtained in closed-form by setting the gradient of $f\left(\mathbf{W}\right)$ to zero, i.e., $\nabla f(\mathbf{W}) \triangleq \mathbf{X}^{\TT} (\mathbf{X}\mathbf{W}-\mathbf{Y})=\mathbf{0}$.
However, since the local datasets typically contain privacy-sensitive information that cannot be disclosed, solving \eqref{eq:loss} in a centralized manner at the server would be infeasible.

\subsection{Federated Learning Algorithm}
\label{sec:comm-model}

Rather than pursuing the optimal solution, the server leverages FL to train a linear regression model without accessing the raw data at the edge devices, and FedAvg \cite{fedavg} is a classic FL algorithm that contains $K$ communication rounds with identical duration $T$.
In the $k$-th communication round of FedAvg, the edge devices first download the global model (denoted as $\mathbf{W}_{k}$) from the server. The model downloading time is denoted as $t_{\mathrm{D}}$.
Then, each edge device performs a $\tau$-step stochastic gradient descent (SGD) to minimize the loss function defined on its local dataset, i.e., $f_i(\mathbf{W}) \triangleq \frac{1}{2} \|\mathbf{X}^{(i)}\mathbf{W} - \mathbf{Y}^{(i)} \|_{\mathrm{F}}^{2}$.
In other words, the updated local model after each SGD step can be written as $\mathbf{W}_{k,u}^{(i)} = \mathbf{W}_{k,u-1}^{(i)} - \eta_k g_{k,u-1}^{(i)} (\mathbf{W}_{k,u-1}^{(i)}), \forall u\in[\tau]$, where $\mathbf{W}_{k,0}^{(i)} = \mathbf{W}_{k}$, $\eta_k$ is the learning rate, and $g_{k,u-1}^{(i)} ( \mathbf{W}_{k,u-1}^{(i)})$ denotes the stochastic gradient computed on a random subset of the local data with batch size $b$.
Denote $\text{MACR}_{i}$ as the Multiply-Accumulate (MAC) rate of edge device $i$ and $N_\text{MAC}$ as the number of MAC operations required for processing one data sample.
The local model updating time in each communication round is thus given as $t_{\mathrm{C}}^{i}(b) = \frac{\tau b N_\text{MAC}}{\text{MACR}_{i}}$ \cite{xu2019elfish}. 
After local model training, each edge device summarizes its model updates as $\mathbf{\hat{g}}_k^{(i)} \triangleq \sum_{u=0}^{\tau-1} g_{k,u}^{(i)}( \mathbf{W}_{k,u}^{(i)})$, which is uploaded to the server for aggregation.
Accordingly, the uploading time is given as $t_{\mathrm{U},k}^{i} = \frac{M}{R_{\mathrm{U},k}^{i}}$, where $M$ (in bits) is the size of the gradient update and $R_{\mathrm{U},k}^{i}=B \log_{2} (1+|h_{k}^{i}|^2 \frac{P_{i}}{N_{0}})$ is the uplink transmission rate, with $B$, $h_{k}^{i}$, $P_{i}$ and $N_{0}$ denoting the channel bandwidth, channel coefficient, uplink transmit power, and receive noise power, respectively.
We further denote $t_k^i(b) \triangleq t_{\mathrm{D}} + t_{\mathrm{C}}^{i}(b)+ t_{\mathrm{U},k}^{i}$ as the minimum time required by edge device $i$ to complete model training and exchange in the $k$-th communication round.

Due to the time-varying wireless channel conditions, it may happen that some edge devices cannot upload their model updates before the deadline of each communication round.
We use an indicator $\mathbbm{1}_{k}^{(i)} \triangleq \mathbbm{1}\{t_k^i(b) \leq T \} $ to denote the arrival status of edge device $i$ in the $k$-th communication round, where $\mathbbm{1}_{k}^{(i)}=1$ means it successfully uploads the model update and $\mathbbm{1}_{k}^{(i)}=0$ otherwise.
By assuming IID block fading, the probability of successfully receiving the model update from edge device $i$ within time $T$ is given as $p_{i} \triangleq \mathrm{Pr} [ t_k^i(b) \leq T ]$, and the new global model is generated at the server by aggregating the received model updates according to $\mathbf{W}_{k+1}=\mathbf{W}_{k}- \eta_{k} \sum_{i=1}^{N} \mathbf{\hat{g}}_{k}^{(i)} \mathbbm{1}_{k}^{(i)}$.

Nevertheless, the efficiency of the above federated linear regression process suffers due to device heterogeneity.
On one hand, edge devices with large MAC rates and better uplink channel quality can finish their local training and model uploading earlier, and thus they have to wait for the slowest device.
To avoid idling the fast devices, we adapt the batch size according to their wireless channel conditions, i.e., edge device $i$ selects the largest batch size $b_{k}^{i}$ for SGD to meet the round deadline, i.e., $b_{k}^{i} = \max_{b\in[l_i]} b \mathbbm{1}\{t_k^i(b) \leq T \} $. 
On the other hand, the straggling devices may have low arrival probabilities, which makes it difficult to utilize their local data for model training. 
In the next section, we propose a new FL framework for linear regression, named stochastic coded federated learning (SCFL), which adopts the coded data at the server to compensate for the missing gradients.

\section{Stochastic Coded Federated Learning}\label{sec:algorithm}

Since the straggling edge devices may fail in uploading the gradients occasionally, we introduce a coded dataset in the proposed SCFL framework for gradient compensation at the server.
Before training starts, each edge device generates the local coded data and uploads them to the server.
Afterward, both the server and edge devices compute the stochastic gradients, which are periodically aggregated to obtain an unbiased model update.
Fig. \ref{fig:codeddata} provides an overview of the SCFL process with different operations detailed as follows.
\begin{figure*}[!t]
    \centering
    \includegraphics[width=0.8\linewidth]{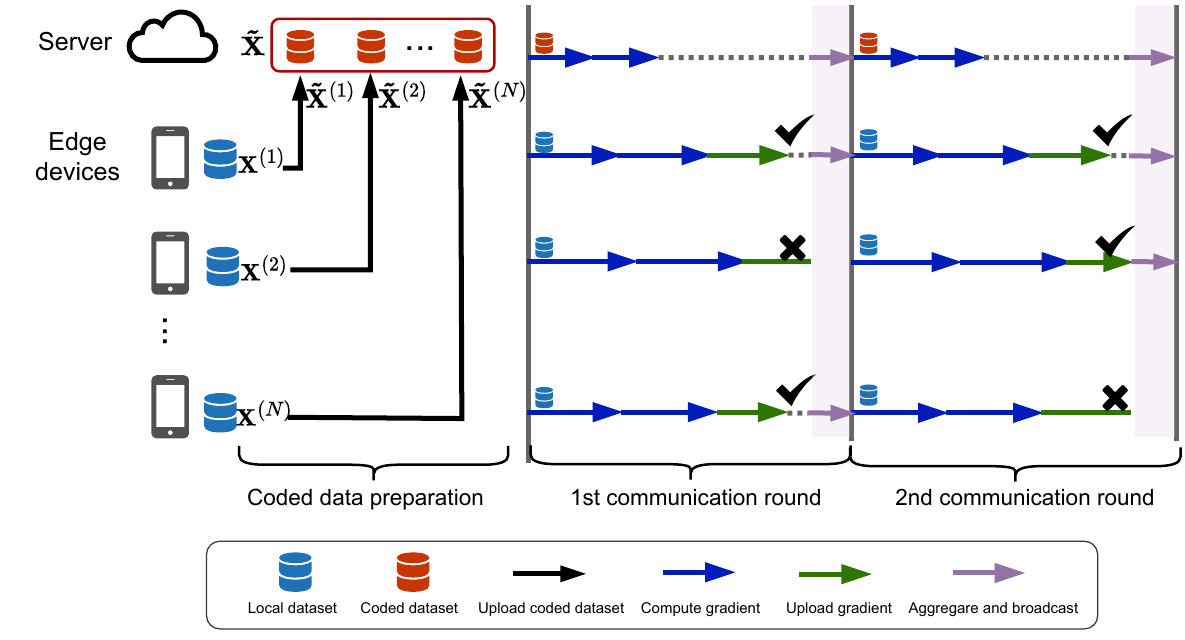}
    \caption{An overview of SCFL. Before training starts, the edge devices upload the coded data to the server. During training, both the server and edge devices perform SGD for multiple steps before gradient aggregation in each communication round. Due to the time-varying wireless channel conditions, some straggling devices may not be able to upload their model updates within time $T$ and the missing gradients are compensated by the coded dataset at the server.}
    \label{fig:codeddata}
\end{figure*}
 
\subsection{Coded Data Preparation}
Before the start of the first communication round, a set of coded data is generated at each edge device via random matrix projection \cite{CFL,CFL_journal}.
However, the privacy protection of simple random projection may be escaped via some knowledge reconstruction attacks \cite{guo2007deriving}.
To provide extra privacy protection, we further add noise to the projected data.
To illustrate the necessity of adding noise, we launch the privacy attack in \cite{guo2007deriving} that is able to reconstruct the original data from the randomly projected data with some prior knowledge. Specifically, with a small subset of the original data samples, one can estimate the projection matrix and use it to decode the original data from the projected data.
The reconstruction error is used as a measure of privacy leakage.
It is intuitive that a larger reconstruction error implies a higher difficulty of reconstructing the original data from the coded data, i.e., data privacy is better preserved.
As shown in Fig. \ref{fig:r1}, coded data with additive noise better preserves data privacy than coded data without adding noise.

Therefore, edge device $i$ computes the coded data features and labels according to $\mathbf{\tilde{X}}^{(i)}=\mathbf{G}_{i} \mathbf{X}^{(i)} + \mathbf{N}_{i}$ and $\mathbf{\tilde{Y}}^{(i)}=\mathbf{G}_{i} \mathbf{Y}^{(i)}$, respectively, where $\mathbf{G}_{i}\in \mathbb{R}^{c \times l_i}$ is the projection matrix with each element independent and identically sampled from the standard Gaussian distribution $\mathcal{N}\left(0, 1\right)$. Also, each entry in $\mathbf{N}_{i} \in \mathbb{R}^{c \times d}$ is independently sampled from $\mathcal{N}\left(0, \sigma_i^2\right)$, where $\sigma_i > 0$ denotes the noise level and $c$ is the number of coded data that is the same at different edge devices.
The local coded data of all the edge devices are collected by the server to construct a global coded dataset as follows:
\begin{equation}
    \mathbf{\tilde{X}} \triangleq \sum_{i=1}^{N} \mathbf{\tilde{X}}^{(i)} = \mathbf{G} \mathbf{X} + \mathbf{N}, \quad
    \mathbf{\tilde{Y}} \triangleq \sum_{i=1}^{N} \mathbf{\tilde{Y}}^{(i)} = \mathbf{G} \mathbf{Y},
\label{eq:encoding}
\end{equation}
where $\mathbf{G} \triangleq [\mathbf{G}_1, \mathbf{G}_2,\dots,\mathbf{G}_N]$ and $\mathbf{N} \triangleq \sum_{i=1}^{N} \mathbf{N}_i$.
Since the server has no access to matrices $\{\mathbf{G}_i\}$'s and $\{\mathbf{N}_i\}$'s, it cannot directly decode the real datasets $\{\mathbf{X}^{(i)}, \mathbf{Y}^{(i)}\}$.
It is worthwhile noting that the noise levels of the coded datasets have direct impacts on both the learning performance and privacy guarantee. Intuitively, adding stronger noise reinforces the privacy protection of local data but degrades the learning performance.
We will characterize the trade-off between the convergence and privacy guarantee using the MI-DP metric in Section \ref{sec:privacy}, and the noise levels at different edge devices will be determined via an incentive mechanism in Section \ref{sec:incentive}.

\subsection{Gradient Computation at the Edge Devices}

In the $k$-th communication round, the local model is initialized as $\mathbf{W}_{k,0}^{(i)} = \mathbf{W}_{k}^{(i)}$ at the $i$-th edge device.
To avoid frequent communication between the server and edge devices, we adopt the periodical averaging approach \cite{fedavg,stich2018local} where the edge devices compute stochastic gradients for $\tau$ steps before uploading the model updates in each communication round. This generalizes the existing works on CFL \cite{CFL,CFL_journal,scfl}, where the model updates need to be uploaded after every local training step.
Specifically, in step $u\in[\tau]$, edge device $i$ updates local model $\mathbf{W}_{k,u}^{(i)}$ by sampling a batch of its local data samples with size $b_{k}^{i}$ for gradient computation.
The sampling matrix is denoted as a random diagonal matrix $\mathbf{S}_{k,u}^{(i)} \in \mathbb{R}^{l_i\times l_i}$ where each diagonal element is independently sampled from the Bernoulli distribution $\mathrm{Bern}(\frac{b_k^i}{l_i})$, i.e., the $j$-th diagonal element of $\mathbf{S}_{k,u}^{(i)}$ equals $1$ if the $j$-th local data sample is sampled.
Accordingly, the batch of randomly sampled data can be denoted as $\left( \mathbf{\hat{X}}_{k,u}^{(i)}, \mathbf{\hat{Y}}_{k,u}^{(i)} \right) \triangleq \left( \mathbf{S}_{k,u}^{(i)} \mathbf{X}^{(i)}, \mathbf{S}_{k,u}^{(i)} \mathbf{Y}^{(i)} \right)$, and the stochastic gradient is computed as follows:
\begin{equation}
    g_{k,u}^{(i)} ( \mathbf{W}_{k,u}^{(i)}) \triangleq \frac{l_i}{b_k^i}\mathbf{\hat{X}}_{k,u}^{(i)\TT} \left( \mathbf{\hat{X}}_{k,u}^{(i)} \mathbf{W}_{k,u}^{(i)} - \mathbf{\hat{Y}}_{k,u}^{(i)} \right),
\label{eq:device}
\end{equation}
where $\mathbf{W}_{k,u}^{(i)} = \mathbf{W}_{k,u-1}^{(i)} - \eta_k g_{k,u-1}^{(i)} ( \mathbf{W}_{k,u-1}^{(i)}), u=1,2,\dots,\tau$.
After $\tau$ local steps, each edge device summarizes the accumulated local gradient as $\mathbf{\hat{g}}_k^{(i)} \triangleq \sum_{u=0}^{\tau-1} g_{k,u}^{(i)}( \mathbf{W}_{k,u}^{(i)})$ and uploads it to the server for aggregation.

\subsection{Gradient Computation at the Server}

To combat the stragglers, the server computes the stochastic gradients on the coded dataset in a similar way as the edge devices but with a higher MAC rate.
Specifically, in the $u$-th step of the $k$-th communication round, it randomly selects $b_{\s}$ data samples via a diagonal sampling matrix $\mathbf{S}_{k,u}^{\s} \in \mathbb{R}^{c\times c}$, where the $j$-th diagonal element is sampled from the Bernoulli distribution $\mathrm{Bern}\left(\frac{b_{\s}}{c}\right)$ and equals $1$ if the $j$-th coded data sample is selected.
Accordingly, the sampled coded data at the server are denoted as $\left( \mathbf{\hat{X}}_{k,u}^{\s},\mathbf{\hat{Y}}_{k,u}^{\s} \right) \triangleq \left( \mathbf{S}_{k,u}^{\s} \mathbf{\tilde{X}},  \mathbf{S}_{k,u}^{\s} \mathbf{\tilde{Y}} \right)$.
To make full use of its MAC rate, the server selects the largest batch size $b_{\s}$ that meets the round deadline.
To eliminate the biased gradient estimate caused by the added noise, we propose a make-up term $g_{\mathrm{o}}(\mathbf{W}_{k,u}^{\s}) \triangleq - \sigma^2 \mathbf{W}_{k,u}^{\s}$, where $\sigma^2 \triangleq \sum_{i=1}^{N} \sigma_i^2$ and $\mathbf{W}_{k,u}^{\s}$ denotes the model at the server, to compute the stochastic gradient as follows:
\begin{equation}
    g_{k,u}^{\s} ( \mathbf{W}_{k,u}^{\s}) \triangleq \frac{1}{b_{\s}} (\mathbf{\hat{X}}_{k,u}^{\s})^{\TT} \left( \mathbf{\hat{X}}_{k,u}^{\s} \mathbf{W}_{k,u}^{\s} - \mathbf{\hat{Y}}_{k,u}^{\s} \right),
\label{eq:server}
\end{equation}
where $\mathbf{W}_{k,u}^{\s} = \mathbf{W}_{k,u-1}^{\s} - \eta_k \left( g_{k,u-1}^{\s} ( \mathbf{W}_{k,u-1}^{\s}) + g_{\mathrm{o}}(\mathbf{W}_{k,u-1}^{\s}) \right)$ and $\mathbf{W}_{k,0}^{\s} = \mathbf{W}_{k}$.
Thus, the accumulated update computed by the server in the $k$-th communication round is summarized as $\mathbf{\hat{g}}_{k}^{\s} \triangleq \sum_{u=0}^{\tau-1} g_{k,u}^{\s}( \mathbf{W}_{k,u}^{\s}) + g_{\mathrm{o}}(\mathbf{W}_{k,u}^{\s})$.

\subsection{Gradient Aggregation Scheme}

Since each communication round has a fixed duration of $T$, only a subset of the edge devices with high MAC rates and favorable channel quality can upload their model updates successfully.
In this regard, if the model updates are aggregated with the same weights, it will be biased towards the local data at fast edge devices, while the training data at stragglers are rarely exploited due to the lower arrival probability.
To resolve this issue, we design a new gradient aggregation scheme by attaching the aggregation weight as the reciprocal of their arrival probability $\frac{1}{p_i}$.
Besides, the missing model updates from stragglers in each communication round are compensated by using the stochastic gradients over the coded data computed at the server.
The global model update in the $k$-th communication round of SCFL can thus be expressed as follows:
\begin{equation}
    g(\mathbf{W}_{k})
    = \frac{1}{2} \left( \sum_{i=1}^{N} \frac{1}{p_i} \mathbf{\hat{g}}_{k}^{(i)} \mathbbm{1}_{k}^{(i)}
    + \mathbf{\hat{g}}_{k}^{\s} \right),
\label{eq:aggregated}
\end{equation}
where the gradient uploaded from edge devices and the coded gradient are given the same weight of $\frac{1}{2}$ in aggregation. This is because they are both equivalent to the gradient computed over the global training dataset in the expected sense.
As will be shown in the next section, $g(\mathbf{W}_{k})$ is an unbiased gradient estimate of the model update over the global dataset.
After the gradient aggregation, the server updates the global model as $\mathbf{W}_{k+1}=\mathbf{W}_{k} - \eta_{k} g(\mathbf{W}_{k})$ and transmits the new global model to the edge devices for the next communication round. 
It also maintains $\mathbf{W}_{k}$ to derive the final learned model as $\mathbf{W}_{\text{SCFL},K} = \frac{1}{\sum_{k=0}^{K-1} \eta_k} \sum_{k=0}^{K-1} \eta_k \mathbf{W}_k$ after $K$ communication rounds.
The complete training process of SCFL is summarized in Algorithm \ref{alg}.

\begin{figure}[t]
    \centering
    \includegraphics[width=0.33\textwidth]{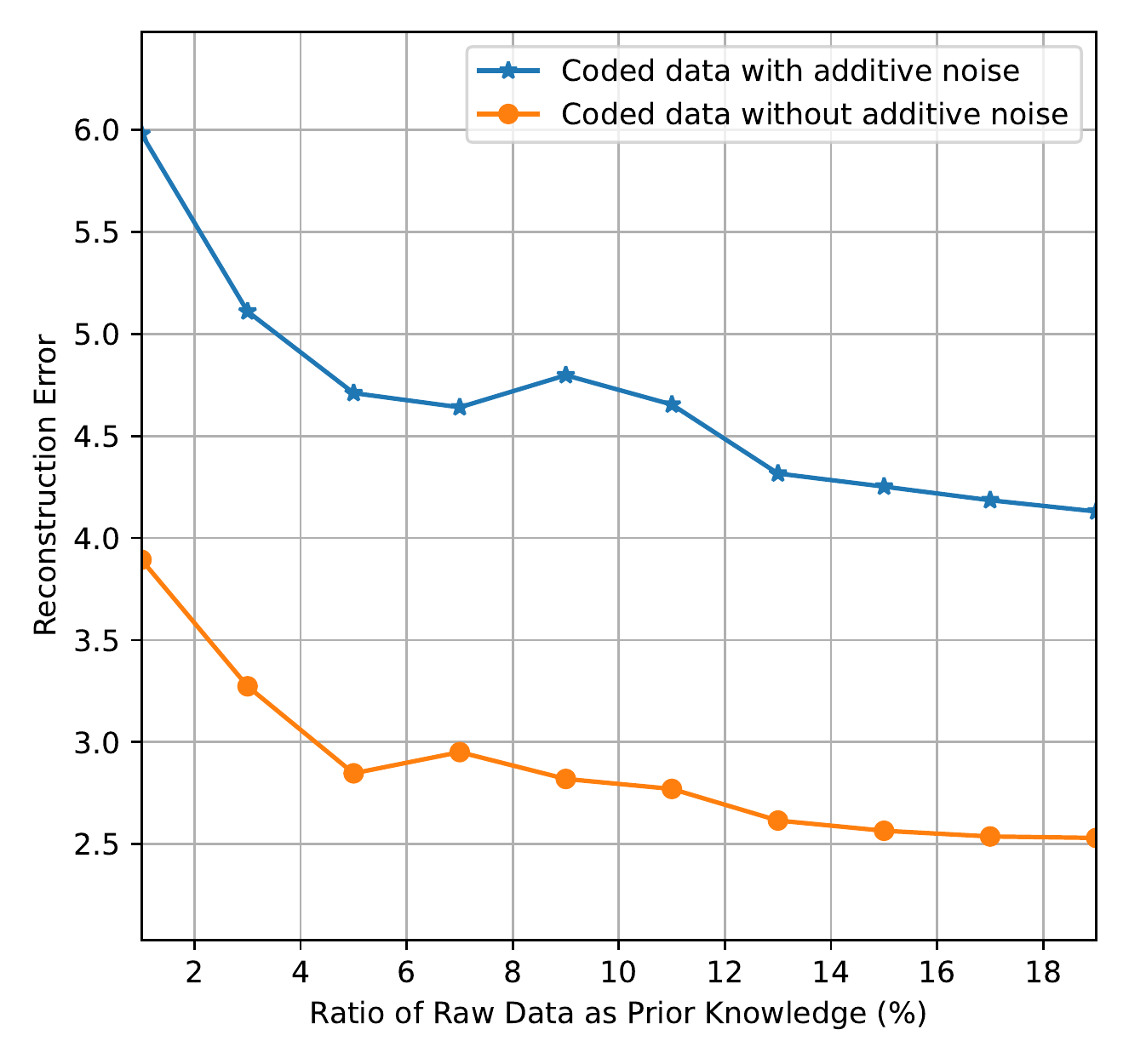}
    \caption{Reconstruction error using the coded data with and without additive noise vs. ratio of raw data samples as prior knowledge at an attacker.}
    \label{fig:r1}
\end{figure}

\textbf{Communication overhead.}
Coded data uploading is one-shot and with a communication cost of $\mathcal{O}(cd)$, while iteratively uploading gradients incurs a communication cost of $\mathcal{O}(Kdo)$.
Although uploading coded data incurs a non-negligible communication cost, it effectively mitigates the straggling effect and improves the learned model accuracy, as will be shown in Section \ref{sec:experiment}.

\begin{algorithm}[tb]
\setlength\abovedisplayskip{0.1cm}
\setlength\belowdisplayskip{0.1cm}
\begin{small}
\caption{Training Process of SCFL} \label{alg}
{\textit{// Before Training Starts //}}
\begin{algorithmic}[1]
\STATE{Each edge device computes the local coded dataset according to $\mathbf{\tilde{X}}^{(i)} = \mathbf{G}_{i} \mathbf{X}^{(i)}+ \mathbf{N}_{i}$ and $\mathbf{\tilde{Y}}^{(i)}=\mathbf{G}_{i} \mathbf{Y}^{(i)}$ and uploads them to the server;}
\STATE{The server constructs the global coded dataset as $(\mathbf{\tilde{X}}, \mathbf{\tilde{Y}}) = (\sum_{i=1}^{N} \mathbf{\tilde{X}}^{(i)}, \sum_{i=1}^{N} \mathbf{\tilde{Y}}^{(i)})$;}
\end{algorithmic}
{\textit{// During Training //}}
\begin{algorithmic}[1]
\STATE{\textbf{Training process at the server:}}
\STATE{Initialize a random global model $\mathbf{W}_{0}$;}
\FOR{$k=0,1,\dots,K-1$}
    \STATE{Broadcast $\mathbf{W}_{k}$ to each edge device $i\in[N]$ and inform them to perform \textit{LocalTrain}($\mathbf{W}_{k}$);}
    \FOR{$u=1,2,\dots,\tau$}
        \STATE{Compute the stochastic gradient based on the coded data according to \eqref{eq:server} and update the model according to $\mathbf{W}_{k,u}^{\s} = \mathbf{W}_{k,u-1}^{\s} - \eta_k \left( g_{k,u-1}^{\s} ( \mathbf{W}_{k,u-1}^{\s}) + g_{\mathrm{o}}(\mathbf{W}_{k,u-1}^{\s}) \right)$;}
    \ENDFOR
	\STATE{Aggregate the computed and received gradients according to \eqref{eq:aggregated} to obtain $g(\mathbf{W}_k)$;}
	\STATE{Update the global model as $\mathbf{W}_{k+1} = \mathbf{W}_{k} - \eta_k g(\mathbf{W}_k)$ and keep a local copy of $\mathbf{W}_{k+1}$;}
\ENDFOR
\RETURN{$\mathbf{W}_{\text{SCFL},K} = \frac{1}{\sum_{k=0}^{K-1} \eta_k} \sum_{k=0}^{K-1} \eta_k \mathbf{W}_k$ as the final learned model}
\STATE{\textbf{Training process at edge device $i$:}}
\STATE{\textbf{def} \textit{LocalTrain}($\mathbf{W}_{k}$):}
\STATE{Initialize $\mathbf{W}_{k,0}^{(i)} = \mathbf{W}_{k}$;}
\FOR{$u=1,2,\dots,\tau$}
	\STATE{Compute the stochastic gradient based on the local data according to \eqref{eq:device};}
	\STATE{Update the local model according to $\mathbf{W}_{k,u}^{(i)} = \mathbf{W}_{k,u-1}^{(i)} - \eta_k g_{k,u-1}^{(i)} ( \mathbf{W}_{k,u-1}^{(i)})$;}
\ENDFOR
\STATE{Compute the local update as $\mathbf{\hat{g}}_k^{(i)} = \sum_{u=0}^{\tau-1} g_{k,u}^{(i)}( \mathbf{W}_{k,u}^{(i)})$;}
\RETURN{$\mathbf{\hat{g}}_{k}^{(i)}$}
\end{algorithmic}
\end{small}
\end{algorithm}

\section{Theoretical Analysis}\label{sec:theory}

In this section, we characterize the convergence performance of SCFL and quantify the privacy leakage in coded data sharing.
Based on these results, we derive a privacy-performance trade-off which is determined by the noise levels of the coded data.

\subsection{Convergence Analysis} \label{sec:convergence}

We first present the following assumption to facilitate the convergence analysis \cite{CFL_journal,hu2016convergence}.

\begin{assumption}\label{ass:norm}
There exist positive constants $\{\alpha_i\}_{i=1}^{N}$, $\{\zeta_i\}_{i=1}^{N}$, $\{\kappa_i\}_{i=1}^{N}$, and $\phi$ such that $\alpha_i^2 \leq \left\| \mathbf{X}^{(i)} \right\|_{\mathrm{F}}^2 \leq \zeta_i^2$, $\left\| \mathbf{X}^{(i)} \mathbf{W} - \mathbf{Y}^{(i)} \right\|_{\mathrm{F}}^2 \leq \kappa_i^2$, and $\left\| \mathbf{W} \right\|_{\mathrm{F}}^2 \leq \phi^2$.
\end{assumption}

According to the definition of matrices $\mathbf{G}$, $\mathbf{N}$, $\{\mathbf{S}_{k,u}^{\s}\}$'s, and $\{ \mathbf{S}_{k,u}^{(i)} \}$'s, we derive some important properties in the following lemma, which are helpful for the analysis.

\begin{lemma}\label{useful-lemma}
Matrices $\mathbf{G}$, $\mathbf{N}$, $\{\mathbf{S}_{k,u}^{\s}\}$'s, and $\{ \mathbf{S}_{k,u}^{(i)} \}$'s have the following properties:
\begin{itemize} 
    \item $\mathbb{E}\left[ \left\|\frac{1}{c}\mathbf{G}^{\TT}\mathbf{G} \right\|_{\mathrm{F}} \right] = \mathbf{I}_m$ and $\mathbb{E} \left[\left\| \frac{1}{c}\mathbf{G}^{\TT}\mathbf{G} - \mathbf{I}_m \right\|_{\mathrm{F}}^2 \right] = \frac{m+m^2}{c}$.
    \item $\mathbb{E}\left[ \left\|\frac{1}{c}\mathbf{N}^{\TT}\mathbf{N} \right\|_{\mathrm{F}} \right] = \sigma^2 \mathbf{I}_d$ and $\mathbb{E} \left[\left\| \frac{1}{c}\mathbf{N}^{\TT}\mathbf{N} - \sigma^2 \mathbf{I}_d \right\|_{\mathrm{F}}^2 \right] = \frac{d+d^2}{c} \sum_{i=1}^{N} \sigma_i^4 $.
    \item $\mathbb{E}\left[ \frac{c}{b_{\s}} (\mathbf{S}_{k,u}^{\s})^{\TT} \mathbf{S}_{k,u}^{\s} \right] = \mathbf{I}_c$ and $\mathbb{E} \left[ \left\| \frac{c}{b_{\s}} (\mathbf{S}_{k,u}^{\s})^{\TT}\mathbf{S}_{k,u}^{\s} - \mathbf{I}_c \right\|_{\mathrm{F}}^2 \right] = \frac{c(c-b_{\s})}{b_{\s}}$.
    \item $\mathbb{E}\left[ \frac{l_i}{b_k^i} {\mathbf{S}_{k,u}^{(i)\TT}}\mathbf{S}_{k,u}^{(i)} \right]  =  \mathbf{I}_{l_i}$ and $\mathbb{E} \left[ \left\| \frac{l_i}{b_k^i}{\mathbf{S}_{k,u}^{(i)\TT}} \mathbf{S}_{k,u}^{(i)} - \mathbf{I}_{l_i} \right\|_{\mathrm{F}}^2 \right]  =  \frac{l_i(l_i-b_k^i)}{b_k^i}, \forall i\in[N]$.
\end{itemize}
\end{lemma}
\begin{proof}
Since each entry of $\mathbf{G}$ is independently sampled from Gaussian distribution $\mathcal{N}\left(0, 1\right)$, $\mathbf{G}^{\TT}\mathbf{G}$ follows the Wishart distribution, i.e., $\mathbf{G}^{\TT} \mathbf{G} \sim W_m(c,\mathbf{I}_m)$, which leads to: $\mathbb{E}[\frac{1}{c} \mathbf{G}^{\TT} \mathbf{G}] = \mathbf{I}_m$ and $\mathbb{E} \left[\left\| \frac{1}{c}\mathbf{G}^{\TT}\mathbf{G} - \mathbf{I}_m \right\|_{\mathrm{F}}^2 \right] = \frac{m+m^2}{c}$.
Besides, since matrix $\mathbf{S}_{k,u}^{\s}$ is symmetric and diagonal, each entry of $\frac{c}{b_{\s}} (\mathbf{S}_{k,u}^{\s})^{\TT} \mathbf{S}_{k,u}^{\s}$ has a unit mean.
Thus, we have $\mathbb{E}\left[ \frac{c}{b_{\s}} (\mathbf{S}_{k,u}^{\s})^{\TT} \mathbf{S}_{k,u}^{\s} \right] = \mathbf{I}_c$ and $\mathbb{E} \left[ \left\| \frac{c}{b_{\s}} (\mathbf{S}_{k,u}^{\s})^{\TT}\mathbf{S}_{k,u}^{\s} - \mathbf{I}_c \right\|_{\mathrm{F}}^2 \right] = \frac{c(c-b_{\s})}{b_{\s}}$.
The proofs for $\mathbf{N}$ and $\{\mathbf{S}_{k,u}^{(i)}\}$'s are similar with those of $\mathbf{G}$ and $\{\mathbf{S}_{k,u}^{\s}\}$, respectively, which are omitted for brevity.
\end{proof}

We define an auxiliary model updating process initialized with global model $\mathbf{W}_{k,0} = \mathbf{W}_{k}$ at the beginning of the $k$-th communication round. 
Assume the server has access to the global dataset,  and it performs $\tau$-step gradient descent according to $\mathbf{W}_{k,u} = \mathbf{W}_{k,u-1} - \eta_{k} \nabla f(\mathbf{W}_{k,u-1}), u= 1,2,\dots,\tau$. The virtual accumulated model update can be summarized as $\mathbf{u}_{k} \triangleq \sum_{u=0}^{\tau-1} \nabla f(\mathbf{W}_{k,u})$.
The main challenge of the convergence analysis for SCFL is to show the aggregated update in \eqref{eq:aggregated} is an unbiased estimate of the virtual accumulated model update $\mathbf{u}_{k}$ with bounded variance, as elaborated in the following lemmas.

\begin{lemma}\label{lem:unbiased}
The global model update in SCFL is an unbiased estimate of the virtual accumulated model update, i.e., $\mathbb{E}[g(\mathbf{W}_{k})]= \mathbf{u}_{k}, \forall k\in[K]$.
\end{lemma}
\begin{proof}
By mathematical induction, we show that in each local step $u\in[\tau]$, the sum of the stochastic gradients over the local dataset, i.e., $\sum_{i=1}^{N} \frac{1}{p_i} g_{k,u}^{(i)}(\mathbf{W}_{k,u}^{(i)}) \mathbbm{1}_{k}^{(i)}$, and the global coded dataset, i.e., $g_{k,u}^{\s}(\mathbf{W}_{k,u}^{\s})$, is an unbiased estimate of the virtual gradient computed over the global dataset, i.e., $\nabla f(\mathbf{W}_{k,u})$.
Then we sum up the gradients over $u=0,1,\dots,\tau \!-\! 1$ to complete the proof.
The detailed proof is delegated to Appendix \ref{proof:unbiased}.
\end{proof}

To show the bounded variance property, we decompose the estimation error of the global update as follows:
\begin{align}
    & \mathbb{E} [\| g(\mathbf{W}_{k}) - \mathbf{u}_{k} \|_{\mathrm{F}}^2 ] \\
    =& \mathbb{E} \Big[ \Big\| \frac{1}{2} \Big( \sum_{u=0}^{\tau-1}  \sum_{i=1}^{N} \frac{1}{p_i} g_{k,u}^{(i)}(\mathbf{W}_{k,u}^{(i)}) \mathbbm{1}_{k}^{(i)} \!-\! \nabla f(\mathbf{W}_{k,u}) \Big) \nonumber \\
    &+ \frac{1}{2} \Big(\sum_{u=0}^{\tau-1}  g_{k,u}^{\s}(\mathbf{W}_{k,u}^{\s}) \!-\! \sigma^2 \mathbf{W}_{k,u}^{\s} \!-\! \nabla f(\mathbf{W}_{k,u}) \Big) \Big\|_{\mathrm{F}}^2 \Big] \nonumber \\
    =& \frac{1}{4} \underbrace{\mathbb{E} \Big[\Big\| \sum_{u=0}^{\tau-1}  \sum_{i=1}^{N} \frac{\mathbbm{1}_{k}^{(i)}}{p_i} g_{k,u}^{(i)}(\mathbf{W}_{k,u}^{(i)}) - \nabla f(\mathbf{W}_{k,u}) \Big\|_{\mathrm{F}}^2 \Big]}_{V_1}  \nonumber \\
    &+ \frac{1}{4}  \underbrace{\mathbb{E} \Big[\Big\| \sum_{u=0}^{\tau-1}  g_{k,u}^{\s}(\mathbf{W}_{k,u}^{\s}) \!-\! \sigma^2 \mathbf{W}_{k,u}^{\s} - \nabla f(\mathbf{W}_{k,u}) \Big\|_{\mathrm{F}}^2 \Big]}_{V_2}, \nonumber
\end{align}
where $V_1$ and $V_2$ respectively represent the gradient estimation errors attributed to the gradients computed on the local datasets and the coded dataset, and the last equality is because of their independence.
Next, we upper bound $V_1$ and $V_2$ in Lemma \ref{lem:device} and Lemma \ref{lem:server}, respectively.

\begin{lemma}\label{lem:device}
The estimation error of the aggregated gradients received from the edge devices is upper bounded as follows:
\begin{equation}
    V_1
    \leq \underbrace{2 \tau \sum_{i=1}^{N} \frac{1 - p_i}{p_i} \zeta_i^2\kappa_i^2 + 
    2\tau \sum_{i=1}^{N} \frac{l_i(l_i-b_k^i)}{b_k^i} \zeta_i^2\kappa_i^2}_{\rho_1}.
\end{equation}
\end{lemma}
\begin{proof}
Please refer to Appendix \ref{appendix-B}.
\end{proof}

\begin{lemma}\label{lem:server}
Define $\mathbf{\tilde{\sigma}} \triangleq [\sigma_1^2,\sigma_2^2,\dots,\sigma_N^2]$.
The estimation error of the gradients computed on the global coded dataset is upper bounded as follows:
\begin{align}
    V_2 \leq & \frac{4\tau}{c} (m+m^2) \zeta\kappa + \frac{4\tau}{c} (d+d^2) \phi^2 \sum_{i=1}^{N} \sigma_i^4 \nonumber \\
    &+ \frac{4 dmn \tau }{c^2} (\zeta\phi^2 + \kappa) \sum_{i=1}^{N} \sigma_i^2 \coloneqq \rho_2(\mathbf{\tilde{\sigma}}).
\end{align}
\end{lemma}
\begin{proof}
Please refer to Appendix \ref{appendix-C}.
\end{proof}

With the above lemmas, we characterize the training performance of SCFL in the following theorem and remark.

\begin{theorem}\label{thm-convergence}
Define the optimality gap after a given training time $T_\text{tot} \triangleq TK $ as $G(T_\text{tot}) \triangleq f \left(\mathbf{W}_{\text{SCFL},K}\right) - f(\mathbf{W}^*)$, where $\mathbf{W}^* = \arg\min_{\mathbf{W}} f(\mathbf{W})$ denotes the optimal linear regression model.
With Assumption \ref{ass:norm}, if the learning rates are chosen as $\eta_{k}L < 1$, we have:
\begin{equation}
    G(T_\text{tot})
    \leq \frac{1-\alpha\eta_0}{2 \sum_{k=0}^{K-1} \eta_k} \left\| \mathbf{W}_0 - \mathbf{W}^* \right\|^2 + \frac{1}{\sum_{k=0}^{K-1} \eta_k} \sum_{k=0}^{K-1} \eta_k^2 \rho (\mathbf{\tilde{\sigma}}).
\label{eq:gap}
\end{equation}
where $\alpha \triangleq \sum_{i=1}^{N}\alpha_i^2 $ and $\rho (\mathbf{\tilde{\sigma}}) \triangleq \frac{1}{4} \rho_1 + \frac{1}{4} \rho_2(\mathbf{\tilde{\sigma}})$.
\end{theorem}
\begin{proof}
It is straightforward to verify that the loss function $f(\cdot)$ is $\alpha$-strongly convex. 
Since the aggregated model update is an unbiased estimate of $\mathbf{u}_{k}$ with bound variance derived in Lemma \ref{lem:device} and Lemma \ref{lem:server}, we can follow a similar proof of Lemma 7 in \cite{woodworth2020local} to conclude the result.
\end{proof}

\begin{remark}\label{rm:convergence}
\textbf{(Convergence Performance)}
If the learning rates further satisfy $\lim_{K\rightarrow \infty} \sum_{k=0}^{K-1} \eta_k = \infty$ and $\lim_{K\rightarrow \infty} \sum_{k=0}^{K-1} \eta_k^2 < \infty$, the right-hand side (RHS) of \eqref{eq:gap} diminishes to zero as $K\rightarrow \infty$, which means the learned model of SCFL, i.e., $\frac{1}{\sum_{k=0}^{K-1} \eta_k} \sum_{k=0}^{K-1} \eta_k \mathbf{W}_k$, converges to the optimal linear regression model.
In particular, by selecting $\eta_{k} = \mathcal{O}\left(\frac{\tau}{(k+\beta)L}\right)$ with some constant $\beta >0$, the dominating term in the RHS of \eqref{eq:gap} becomes $\mathcal{O}\left(\frac{1}{\tau \sum_{k=1}^K (k+\beta)}\right)$.
Besides, the RHS of \eqref{eq:gap} increases with the gradient estimation error $\rho (\mathbf{\tilde{\sigma}})$, which means adding stronger noise $\mathbf{\tilde{\sigma}}$ degrades the performance of the learned model. 

With the selected learning rates $\eta_{k} = \mathcal{O}\left(\frac{\tau}{(k+\beta)L}\right)$, SCFL achieves a convergence rate of $\mathcal{O}(\frac{1}{\tau K})$, i.e., it converges geometrically fast with respect to the number of local steps $\tau$ and communication rounds $K$, which is because of the $\alpha$-strong convexity of loss function $f(\cdot)$ \cite{bottou2018optimization,reisizadeh2020fedpaq}. In comparison, for general convex loss functions, the convergence rate of FL using SGD for local model updates becomes slower, which is typically given as $\mathcal{O}(\frac{1}{\sqrt{\tau K}})$ \cite{scaffold}.

\end{remark}

To obtain a better model, the server prefers collecting less noisy coded datasets, which, however, degrades the privacy protection of local datasets.

\subsection{Privacy-Performance Tradeoff Analysis}\label{sec:privacy}

To characterize the privacy leakage caused by sharing the local coded datasets, we adopt the $\epsilon$-mutual information differential privacy ($\epsilon$-MI-DP) \cite{MI-DP} metric, which captures the privacy risk of individual entries in the dataset.
Compared with original DP definitions \cite{dp,dp2}, MI-DP is often considered to be more intuitive, because it explicitly incorporates the well-established concept of mutual information.
We first introduce the following assumption to meet the definition of $\epsilon$-MI-DP, which can be easily satisfied by normalizing each entry in $\mathbf{X}$. 

\begin{assumption}\label{ass:max_entry}
The maximum absolute value of entries in $\mathbf{X}$ is upper bounded by $1$.
\end{assumption}

\begin{definition}
\textbf{($\epsilon$-MI-DP \cite{MI-DP})} A randomized mechanism $q(\cdot)$ that encodes local data $\mathbf{X}^{(i)}$ to $\mathbf{\tilde{X}}^{(i)}$ satisfies the $\epsilon_{i}$-mutual information differential privacy if
\begin{equation}
    \sup_{j, P_{\mathbf{X}^{(i)}}} I\left(\mathbf{X}^{(i)}[j] ; \mathbf{\tilde{X}}^{(i)} | \mathbf{X}^{(i)}[-j] \right) \leq \epsilon_{i},
    \label{DP-def}
\end{equation}
where the supremum is over the entry index $j$ and distribution $P_{\mathbf{X}^{(i)}}$ of the local dataset $\mathbf{X}^{(i)}$, and $\mathbf{X}^{(i)}[-j] $ denotes the dataset $\mathbf{X}^{(i)}$ excluding the $j$-th sample $\mathbf{X}^{(i)}[j]$.
\end{definition}

In particular, \eqref{DP-def} provides an upper bound for the maximum privacy leakage of each entry in the local dataset.
Notably, a smaller value of the privacy budget $\epsilon_i$ offers better privacy protection.
We select the maximum privacy budget of all the edge devices as the privacy loss caused by coded data sharing, i.e., $\epsilon \triangleq \max_{i\in[N]} \epsilon_{i}$, as stated in the following theorem.

\begin{theorem} \label{theorem:privacy_budget}
With Assumption \ref{ass:max_entry}, the privacy leakage caused by sharing the local coded datasets is given as follows:
\begin{equation}
    \epsilon = \max_{i\in[N]} \left\{ \epsilon_i= q_i(\sigma_i^2) \triangleq \frac{1}{2} \log_{2} \bigg( 1+\frac{c}{h^{2}\big(\mathbf{\tilde{X}}^{(i)}\big)+ \sigma_i^2} \bigg) \right\},
    \label{eq:DP}
\end{equation}
where $h(\mathbf{\tilde{X}}^{(i)}) \triangleq \min\limits_{k_{2}} \sqrt{\sum\nolimits_{k_{1}=1}^{l_{i}}|\mathbf{X}_{k_{1},k_{2}}^{(i)}|^{2}  -  \max\limits_{k_{3} \in[l_{i}]}|\mathbf{X}_{k_{3},k_{2}}^{(i)}|^{2}}$ with the $\mathbf{X}_{j,k}^{(i)}$ denoting the $(j,k)$-th entry of matrix $\mathbf{X}^{(i)}$.
\end{theorem}
\begin{proof}
The proof is available at Lemma 2 of \cite{Privacy_utility_tradeoff} and is omitted for brevity.
\end{proof}

With the result in Theorem \ref{theorem:privacy_budget}, we first compare the privacy protection of different CFL schemes and then show the privacy-performance tradeoff, as stated in the following remarks, respectively.

\begin{remark}\label{rm:privacy}
The coded datasets in CFL-FB \cite{CFL} and CodedFedL \cite{CFL_journal} are constructed as the random linear projection of the local data, and thus their privacy leakage can be viewed as a special case of Theorem \ref{theorem:privacy_budget} with $\sigma_i=0, \forall i\in[N]$.
By adding Gaussian noise to the coded data, the proposed SCFL framework and DP-CFL \cite{anand2021differentially} provides better privacy protection (i.e., a smaller $\epsilon$) than CFL-FB and CodedFedL.
\end{remark}

\begin{remark}\label{rm:tradeoff}
\textbf{(Privacy-performance tradeoff)}
According to Theorems \ref{thm-convergence} and \ref{theorem:privacy_budget}, there is a tradeoff between privacy protection and convergence performance of SCFL.
Specifically, by adding stronger noise $\sigma_{i}^2$ to the local coded dataset, each edge device can meet the requirement of a smaller privacy budget, which reduces the privacy leakage in coded dataset sharing.
However, it also degrades the training performance, as the gradient estimation error $\rho (\mathbf{\tilde{\sigma}})$ increases with each added noise $\sigma_{i}^2$.
\end{remark}

This privacy-performance tradeoff shows a clear conflict between convergence performance and privacy leakage.
In short, each edge device has its preference for privacy protection, while the server expects to receive less noisy coded datasets to improve the training performance.
To solve this conflict, we next develop an incentive mechanism to determine the proper noise levels.

\section{An Incentive Mechanism via Contract Theory} \label{sec:incentive}
In this section, we propose a contract-based incentive mechanism for SCFL, which determines the mutually satisfactory noise levels for the local coded datasets.

\subsection{Problem Formulation}

To improve the learning performance, Theorem \ref{thm-convergence} indicates that it is beneficial to reduce the gradient variance at the server, which can be achieved by gathering less noisy local coded datasets when constructing the global coded dataset.
However, it is almost impossible to obtain an exact and tractable characterization of the impacts of the added noise level on the learning performance since many parameters in the convergence bound are unknown. 
Therefore, we model the effect of the added noise level $\sigma_i^2$ on the learning performance as a general non-increasing concave function $\Gamma(\sigma_i^2)$ \cite{zhan2020learning}.
As our target is to provide a feasible mechanism to derive the mutually satisfactory noise levels for the local coded datasets, this utility function is sufficient to abstract the impacts of different noise levels.
To protect the privacy of local dataset, the edge devices tend to add stronger noise to the coded samples.
The server has to pay rewards to the edge devices to motivate them to share more accurate coded data.
We denote the reward paid to the $i$-th edge device by the server as $r_i$. Thus, the utility of the server can be expressed as follows:
\begin{equation}
    U_S \left(\{\sigma_{i}^{2}\},\{r_{i}\}\right) = \sum_{i=1}^{N} \Gamma(\sigma_{i}^{2}) - \lambda \sum_{i=1}^{N} r_i,
\label{eq:u_s_1}
\end{equation}
where $\lambda > 0$ is a weight parameter that can be adjusted by the server to control the total reward within an acceptable budget.
Specifically, a larger value of $\lambda$ implies that the server pays more rewards to obtain a better learning performance.
In practice, we can build a reference table (e.g., Table \ref{table:lambda} in Appendix \ref{sec:table}) that shows the correspondence between the value of $\lambda$ and the total reward once the optimal contract items are derived.
By referring to this table, the server can choose the value of $\lambda$ according to its reward budget.

The edge devices expect to have lower privacy costs while receiving larger rewards.
We model the privacy cost of edge device $i$ as $\mu_i\epsilon_i$, where $\mu_i > 0$ denotes its privacy sensitivity \cite{ghosh2011selling,ng2021hierarchical,liu2021privacy}.
Specifically, with a larger value of $\mu_i$, edge device $i$ concerns more with the risk of privacy leakage and thus aspires to receive more rewards for any given privacy budget.
In practice, the values of $\{\mu_i\}$'s are set by edge devices \cite{spiekermann2001privacy}, which may depend on their locations, data types, and contents.
Without loss of generality, we assume that $\chi = \{\mu_i: i\in[N] \}$ are sorted in an ascending order, i.e., $0 < \mu_1 \leq \mu_2 \leq \dots \leq \mu_N$.
The utility function of the $i$-th edge device can be written as follows:
\begin{equation}
    U_i(\epsilon_i,r_i) = r_i - \mu_i\epsilon_i,
    \forall i\in[N].
\end{equation}

The server needs to design a contract $\Omega(\chi)$ that contains a set of contract items $\{ \epsilon_i, r_i \}_{i=1}^{N}$, to maximize its utility as well as ensure the individual rationality (IR) and incentive compatibility (IC) of the edge devices, which are formally defined below.
\begin{definition}
\textbf{(Individual Rationality (IR))}
The edge devices upload the local coded dataset only when a non-negative utility can be achieved, i.e.,
\begin{equation}
    U_i(\epsilon_i,r_i) \geq 0, \forall i\in[N].
    \label{eq:IR}
\end{equation}
\end{definition}
\begin{definition}
\textbf{(Incentive Compatibility (IC))}
An edge device always adopts a contract item that can achieve the maximal utility, i.e.,
\begin{equation}
    U_i(\epsilon_i,r_i) \geq U_i(\epsilon_{i^{\prime}},r_{i^{\prime}}), \forall i\neq i^{\prime}, i,i^{\prime} \in [N].
    \label{eq:IC}
\end{equation}
\end{definition}

Note that $\sigma_i^2 = q^{-1}_i(\epsilon_{i})$ with $q^{-1}_i(x) \triangleq \frac{c}{2^{2x} - 1} - h^{2}(\mathbf{\tilde{X}}^{(i)})$.
To satisfy the IR and IC requirements, the contract design at the server can be formulated as the following optimization problem:
\begin{align}
    \mathbf{P1}: \max_{\Omega} \; & \tilde{U}_S (\Omega(\chi)) \\
    \text{s.t.} \; & \eqref{eq:IR}, \eqref{eq:IC}, \\
    & r_i \geq 0, \forall i\in[N],\\
    &0\leq \epsilon_i \leq q_i(0), \forall i\in[N], \label{eq:constrant2}
\end{align}
where $ \tilde{U}_S (\Omega(\chi)) \triangleq \sum_{i=1}^{N} \Gamma(q^{-1}_i(\epsilon_{i})) - \lambda \sum_{i=1}^{N} r_i$ and the last constraint follows the noise level is non-negative, i.e., $\sigma_{i}^{2} \geq 0$. 
It is straightforward to verify Problem (P1) is a convex optimization problem. However, instead of solving (P1) via the Karush–Kuhn–Tucker (KKT) conditions, we develop a low-complexity solution in the next subsection by establishing the relationship between the optimal rewards $\{r_i^*\}$ for a given set of privacy budgets $\{\epsilon_i\}$.
Besides, while the constraint \eqref{eq:constrant2} indicates a minimum noise level of zero, the proposed solutions can be easily generalized to the case with a noise level constraint $\sigma_{\text{min}}^2$, i.e., $0 < \epsilon_i \leq q_i(\sigma_{\text{min}}^2), \forall i \in [N]$.

\subsection{Optimal Contract Design}

To solve Problem (P1), we follow the methods in \cite{gao2011spectrum,limcontract} to first find the optimal rewards $\{r_i^*\}_{i=1}^{N}$ given the privacy costs of the edge devices, as stated in the following theorem.
\begin{theorem}\label{thm:optimal-r}
Given $0 < \mu_1 \leq \mu_2 \leq \dots \leq \mu_N$ and $\{\epsilon_{i}\}_{i=1}^{N}$, the optimal rewards are given as follows:
\begin{equation}
    r_{i}^{*} = \left\{
    \begin{array}{ll}
    \mu_{i} \epsilon_i,
    & \text{if }i=N, \\
    r_{i+1}^{*} - \mu_{i} \epsilon_{i+1} + \mu_{i} \epsilon_{i}, & \text{otherwise.}
    \end{array}
    \right.
\end{equation}
\end{theorem}
\begin{proof}
Please refer to Appendix \ref{proof:contract}.
\end{proof}

The above theorem simplifies the objective function in Problem (P1) to $\tilde{U}_S(\Omega(\chi)) = \sum_{i=1}^{N}$ $\Gamma(q^{-1}_i(\epsilon_{i})) - \lambda \sum_{i=1}^{N} [ i \mu_{i} \epsilon_{i} - (i-1) \mu_{i-1} \epsilon_{i} ]$ where only $\{\epsilon_i\}$'s are the optimization variables.
Define $\Phi_i(\epsilon_i) \triangleq \Gamma(q^{-1}_i(\epsilon_{i})) - \lambda i  \mu_{i} \epsilon_i + \lambda (i-1) \mu_{i-1} \epsilon_i $.
Problem (P1) is equivalent to the following optimization problem:
\begin{align}
    \mathbf{P2}: \max_{\{\epsilon_{i}\}_{i=1}^{N}} \; & \sum_{i=1}^{N} \Phi_i(\epsilon_i),  \\
    \text{s.t.} \;& \epsilon_1 \geq \epsilon_2 \geq \dots \geq \epsilon_N > 0, \label{eq:constraint} \\
    & 0 < \epsilon_i \leq q_i(0), \forall i \in [N] \label{eq:constraint2},
\end{align}
where the constraint \eqref{eq:constraint2} indicates the privacy budget should be upper bounded by that corresponds to the case without adding noise to the coded dataset.

Notably, with the concave objective function and the monotonic property of $\{\epsilon_i\}$'s, Problem (P2) can be optimally solved via the Bunching and Ironing Algorithm \cite{gao2011spectrum}.
Specifically, we first find the maximizer $\{\epsilon_{i}^{*}\}$ of the objective function by disregarding the constraints in \eqref{eq:constraint}.
Then, starting from $i = N$, we check whether $\epsilon_{i-1}^* \geq \epsilon_{i}^{*}$ holds.
If this inequality holds, we set $i$ as $i-1$ and proceed.
Otherwise, there must exist some $j \leq i-1$ such that $\epsilon_{j}^{*} \leq \epsilon_{j+1}^{*} \leq \dots < \epsilon_{i}^{*}$ (Line 4 of Algorithm \ref{bunching}), and we update their values by setting $\epsilon_{j}^{*} = \epsilon_{j+1}^{*} = \dots = \epsilon_{i}^{*}$ as $\arg\max_{\epsilon \in [0, \min_{l} Q_{l}(0)]} \sum_{l=j}^{i} \Phi_{l}(\epsilon)$ (Line 5 of Algorithm \ref{bunching}).
This ensures $\epsilon_{j}^{*} \geq \epsilon_{j+1}^{*} \geq \dots \geq \epsilon_{i}^{*}$ so that we can proceed by setting $i=j$.
In the worst case, adjustment is needed for each $i = N, N-1, \dots, 2$, i.e., the \textit{while} loop in Algorithm \ref{bunching} is repeated for at most $N-1$ times.
Details of the Bunching and Ironing Algorithm for Problem (P2) are summarized in Algorithm \ref{bunching}.
It is worth noting that since edge devices are privacy-sensitive, the optimal noise levels  rarely attain the noise-free upper bounds in \eqref{eq:constraint2}.
Furthermore, the Bunching and Ironing Algorithm is also applicable if a minimum noise level $\sigma_{\text{min},i}^2$ is added to the coded data, i.e., $0 < \epsilon_i \leq q_i(\sigma_{\text{min},i}^2), \forall i \in [N]$.

\begin{algorithm}[tb]
\begin{small}
\caption{Bunching and Ironing Algorithm} \label{bunching}
\begin{algorithmic}[1]
\STATE{Initialize $\epsilon_i^{*} = \arg\max_{ \epsilon_i \in (0, q^{-1}_i(0)]} \Phi_i(\epsilon_i), \forall i\in[N]$ and $i=N$;}
\WHILE{$i>1$}
    \IF{$\epsilon_{i-1}^{*} < \epsilon_{i}^{*}$} 
    \STATE{Find the smallest $j\leq i-1$ such that $\epsilon_{j}^{*} \leq \epsilon_{j+1}^{*} \leq \dots < \epsilon_{i}^{*}$;}
    \STATE{Calculate $\epsilon = \arg \max_{\epsilon \in [0, \min_{l} Q_{l}(0)]} \sum_{l = j}^{i} \Phi_{l}(\epsilon)$ and set $\epsilon_{l}^{*} = \epsilon, \forall l = j,j+1 \dots, i$;}
    \ENDIF
    \STATE{Set $i=j$;}
\ENDWHILE
\RETURN{$\{\epsilon_i^{*}\}_{i=1}^{N}$}
\end{algorithmic}
\end{small}
\end{algorithm}

\section{Performance Evaluation}\label{sec:experiment}

\subsection{Experimental Setup}

We consider an FL system with $20$ edge devices.
The channel bandwidth $B$ and noise power $N_0$ are set as $180$ kHz and $-70$ dBm, respectively. 
To simulate different pathloss, the transmit powers of the edge devices are uniformly sampled in the range of $15 \sim 25$ dBm, and the channel gains $\{|h_{k}^{i}|^{2}\}$ follow the exponential distribution with the default mean value $\gamma = 10^{-8}$.
Besides, the downloading rate and the MAC rate of the server are set as $1$ Mbps and $15,360$ KMAC per second, respectively.
We randomly generate the MAC rates of the edge devices according to $\text{MACR}_{i} = \mu_{\text{comp},i} \times 1,536$ KMAC per second, where $\mu_{\text{comp},i}$ is uniformly sampled from $[0.8,1.0]$.
We evaluate the FL algorithms on two image datasets, including the MNIST \cite{mnist} and CIFAR-10 \cite{cifar10} datasets.
In addition, the default value of $\tau$ is $5$, and $T$ is set as $10$ and $15$ seconds for experiments on the MNIST and CIFAR-10 datasets, respectively.

We compare the proposed SCFL scheme with the following three FL benchmarks:
\begin{itemize}
\item \textbf{FedAvg} \cite{fedavg}: 
In each communication round, the edge devices compute the stochastic gradients for $\tau$ local steps and upload the accumulated model updates to the server.
The server then aggregates the received model updates from fast edge devices and generates a new global model.

\item \textbf{CodedFedL} \cite{CFL_journal}: 
In CodedFedL, the local coded dataset is generated by random projection of the local data without adding noise. Both the server and the edge devices divide their datasets into several batches.
In each communication round of the training process, a random batch is selected to compute the gradient at each edge device and the server.
However, the server and edge devices compute the stochastic gradients only for one step (i.e., $\tau=1$) before aggregation. 

\item \textbf{DP-CFL} \cite{anand2021differentially}:
The coded dataset is generated with the same method as the proposed SCFL scheme. However, only the server performs centralized model training using the global coded dataset without further cooperation with the edge devices.

\end{itemize}
The CodedFedL scheme obtains the optimal batch sizes of the data at the server and edge devices by solving the optimization problem stated in (23) of \cite{CFL_journal}. 
In all CFL methods, the number of coded data samples $c$ is set as $10,000$, and the batch sizes of the server $b_{\s}$ are $499$ and $562$ for the MNIST and CIFAR-10 datasets, respectively.

To simulate the non-IID data on edge devices, we sort the images in each dataset by their labels, divide them into 20 shards with identical size, and assign one random shard to each edge device \cite{hsieh2020non}.
Following \cite{CFL_journal}, the classification task on the MNIST dataset is transformed into a linear regression problem by using the random Fourier feature mapping (RFFM) \cite{rahimi2008uniform}.
Accordingly, each transformed vector has a dimension of 2,000.
As for the CIFAR-10 dataset, we use a 4,096-dimensional feature vector extracted by a pretrained-VGG model \cite{vgg} to represent each image.
Therefore, we can train linear regression models for classifying both image datasets.

\begin{figure}[!t]
\centering
\subfigure[MNIST dataset]{
\begin{minipage}[t]{0.5\linewidth}
\centering
\includegraphics[width=\linewidth]{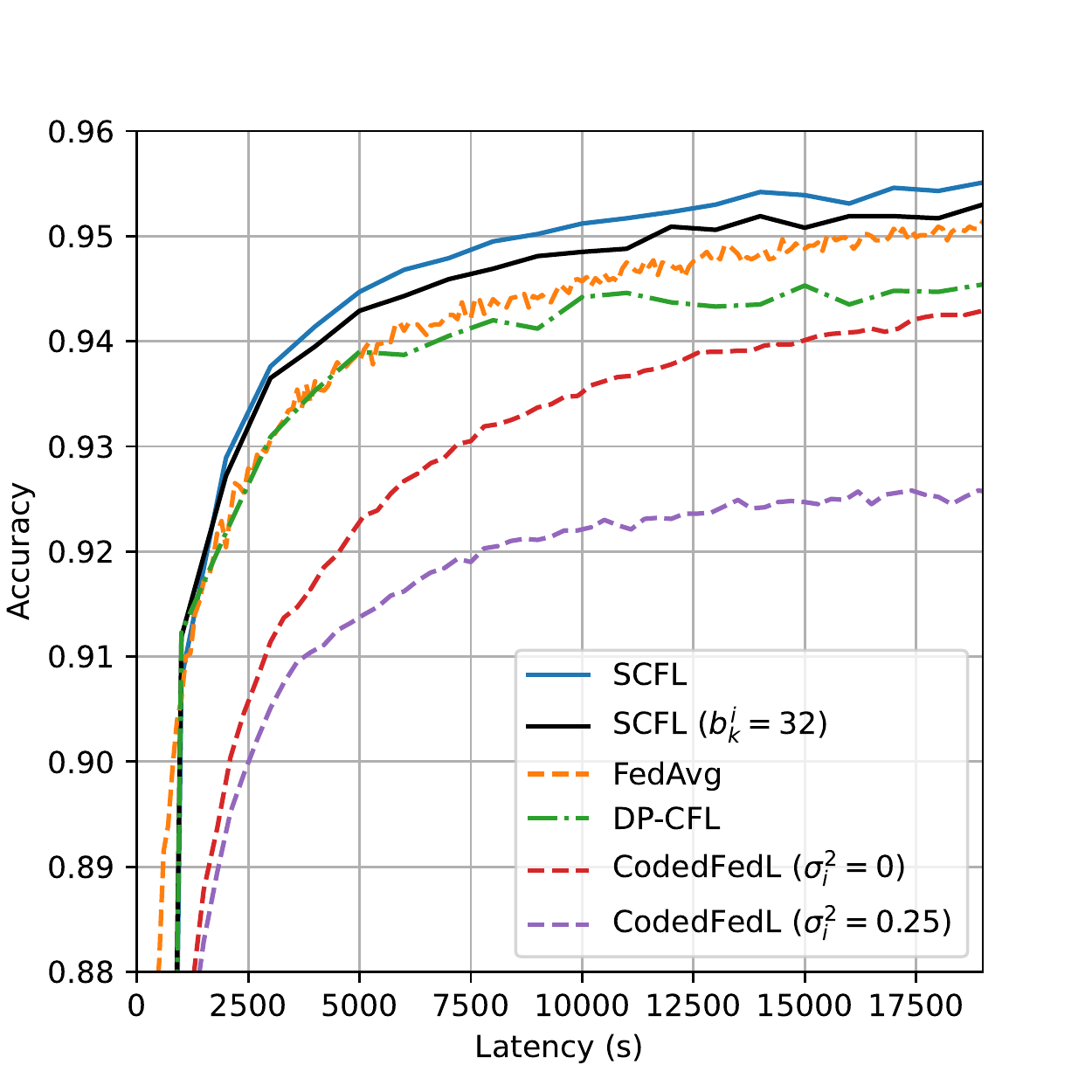}
\label{bsc+mnist}
\end{minipage}%
}%
\subfigure[CIFAR-10 dataset]{
\begin{minipage}[t]{0.5\linewidth}
\centering
\includegraphics[width=\linewidth]{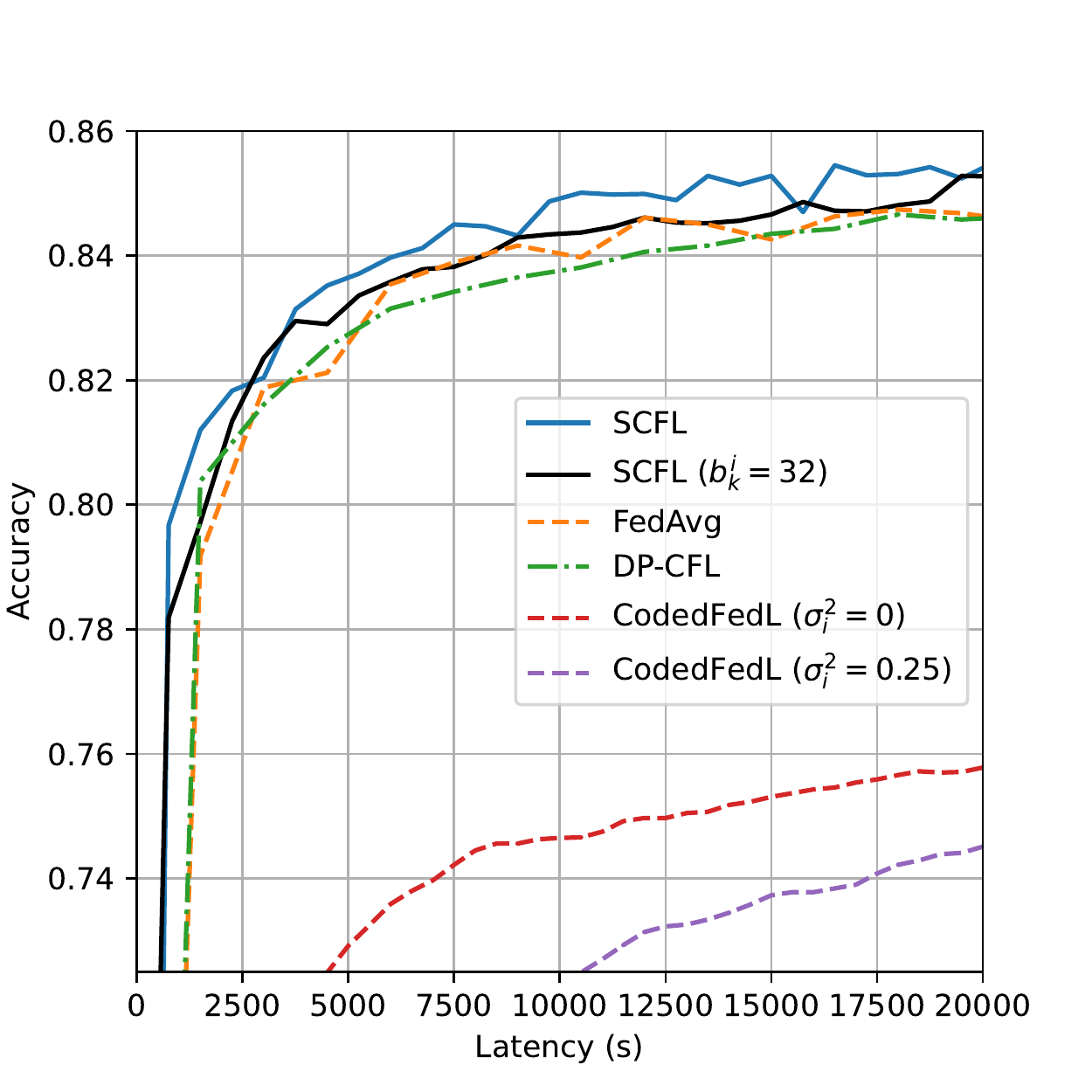}
\label{awgn+cifar10}
\end{minipage}
}%
\centering
\caption{Test accuracy of different FL schemes vs. training latency on (a) the MNIST dataset and (b) the CIFAR-10 dataset. SCFL ($b_{k}^{i}=32$) denotes SCFL with constant batch size $b_k^i = 32, \forall i\in[N], k\in[K]$.}
\label{Fig:convergence}
\end{figure}

\subsection{Results}

\subsubsection{Training Performance Comparison}
We compare the model accuracy achieved by different FL frameworks with respect to the training latency in Fig. \ref{Fig:convergence}, where the noise levels in all the CFL schemes are set as $\sigma_{i}^{2} = 0.25, \forall i \in [N]$, implying the same privacy budgets.
It is observed that CodedFedL achieves the lowest test accuracy due to the frequent communication and model update bias introduced by the noisy coded datasets.
Even without the added noise (i.e., $\sigma_{i}^{2} = 0$), the performance of CodedFedL is still worse than other schemes, which further demonstrates the benefit of periodical averaging in accelerating the convergence speed.
DP-CFL achieves a similar test accuracy with FedAvg, since the server continuously computes the stochastic gradients without any communication overhead.
Compared with the baseline methods, our proposed SCFL achieves the highest test accuracy within the given training time, and the batch adaptation scheme further accelerates the model convergence by exploiting more data samples for training.
This is because the proposed aggregation scheme in SCFL can effectively utilize the global coded dataset to mitigate the stragglers by fully exploiting the computational resources at both the server and edge devices.

\subsubsection{Privacy-Performance Tradeoff}

\begin{figure}[!t]
\centering
\subfigure[MNIST dataset]{
\begin{minipage}[t]{0.5\linewidth}
\centering
\includegraphics[width=\linewidth]{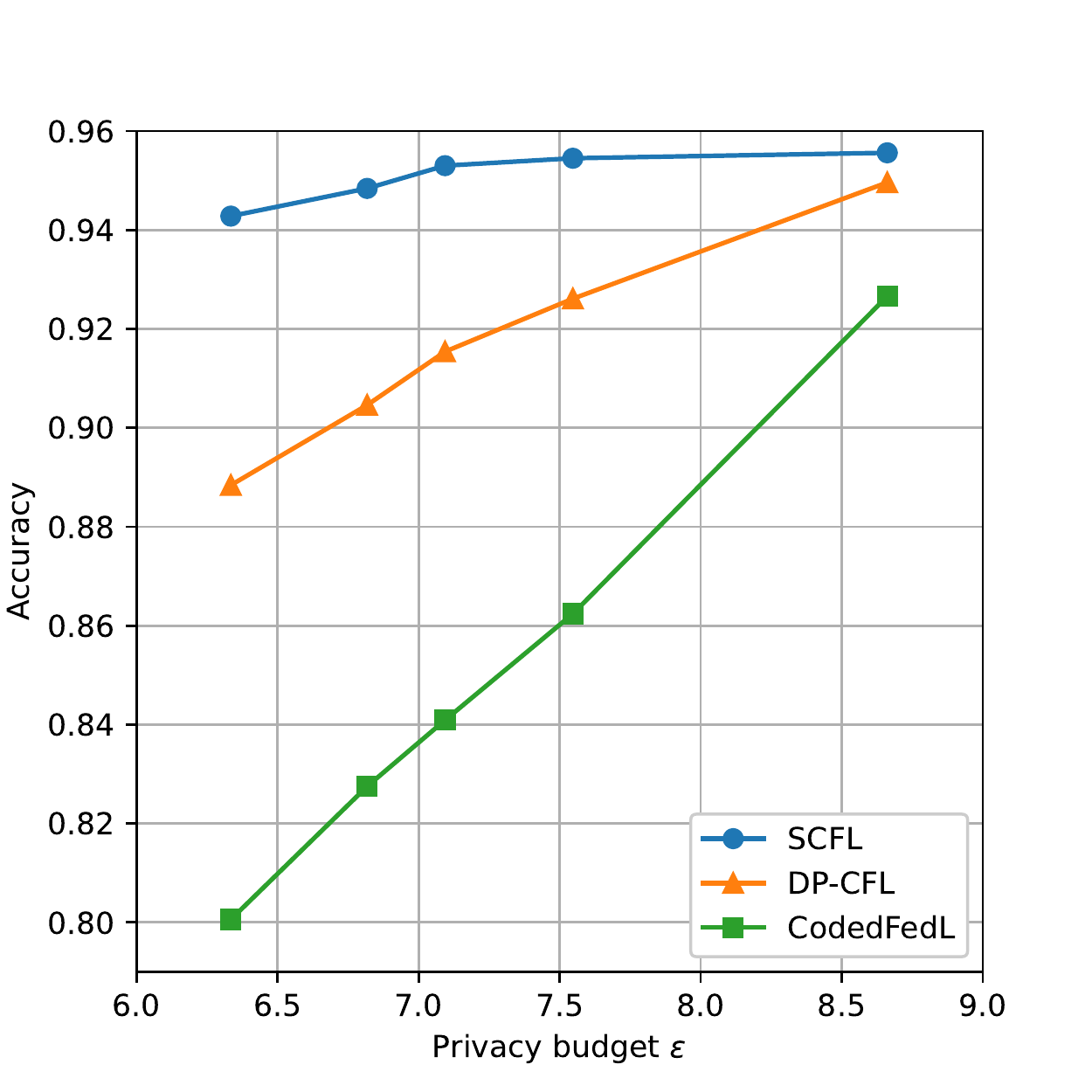}
\label{tradeoff-fig1}
\end{minipage}%
}%
\subfigure[CIFAR-10 dataset]{
\begin{minipage}[t]{0.5\linewidth}
\centering
\includegraphics[width=\linewidth]{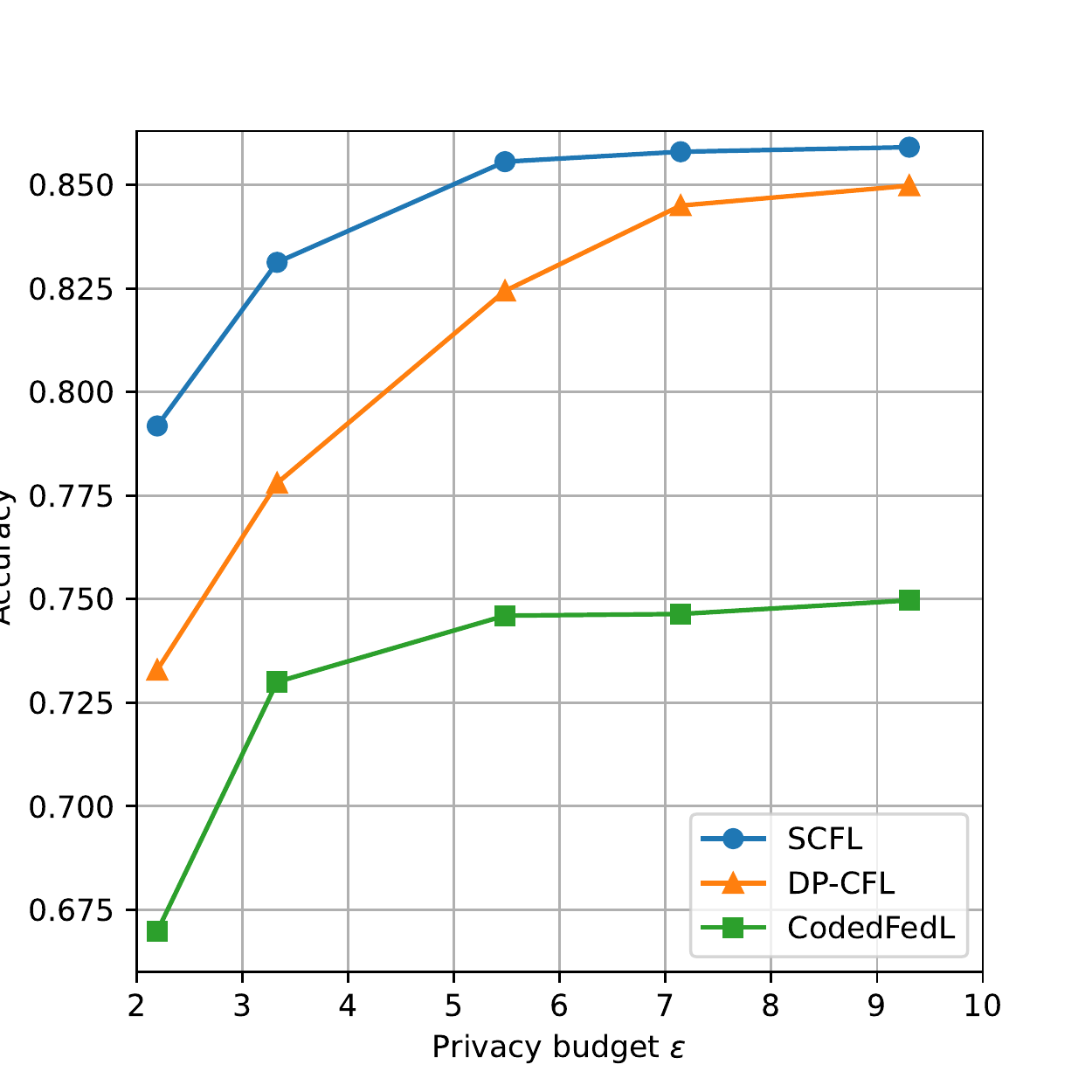}
\label{tradeoff-fig2}
\end{minipage}%
}%
\centering
\caption{Privacy-performance tradeoff on (a) the MNIST dataset and (b) the CIFAR-10 dataset. All the CFL schemes are trained with $T_{tot} = 20,000$ seconds.}
\label{Fig:privacy-performance tradeoff}
\end{figure}

We change the noise levels $\sigma^2$ of the local coded datasets to simulate different privacy budgets and investigate the privacy-performance tradeoff of different CFL schemes in Fig. \ref{Fig:privacy-performance tradeoff}.
It can be observed that the learned model achieves a higher test accuracy by increasing the privacy budget of MI-DP since the coded data are less noisy. This demonstrates the privacy-performance tradeoff highlighted in Remark \ref{rm:tradeoff}.
Among the three CFL schemes, SCFL secures the highest test accuracy under any given privacy budget, i.e., it achieves the best privacy-performance tradeoff.
This is again attributed to the proposed gradient aggregation scheme of SCFL that effectively mitigates the negative effects of the added noise.
Besides, the model obtained by CodedFedL has the worst test accuracy since its efficiency is degraded by the frequent communications. 
Comparing Fig. \ref{tradeoff-fig1} and Fig. \ref{tradeoff-fig2}, while the test accuracies on the MNIST dataset achieved by different CFL schemes keep increasing with the privacy budget, the test accuracies on the CIFAR-10 dataset plateau when the privacy budget is larger than $8.0$.
This is because the MNIST dataset is much easier to classify and therefore more sensitive to the added noise.
It is also interesting to note that under a very low noise level, i.e., $\sigma^2 \leq 5$, the proposed SCFL scheme still outperforms DP-CFL by notable margins since it improves the training efficiency by exploiting the computational resources on both the server and edge devices.

\begin{figure}[!t]
\centering
\subfigure[MNIST dataset]{
\begin{minipage}[t]{0.5\linewidth}
\centering
\includegraphics[width=\linewidth]{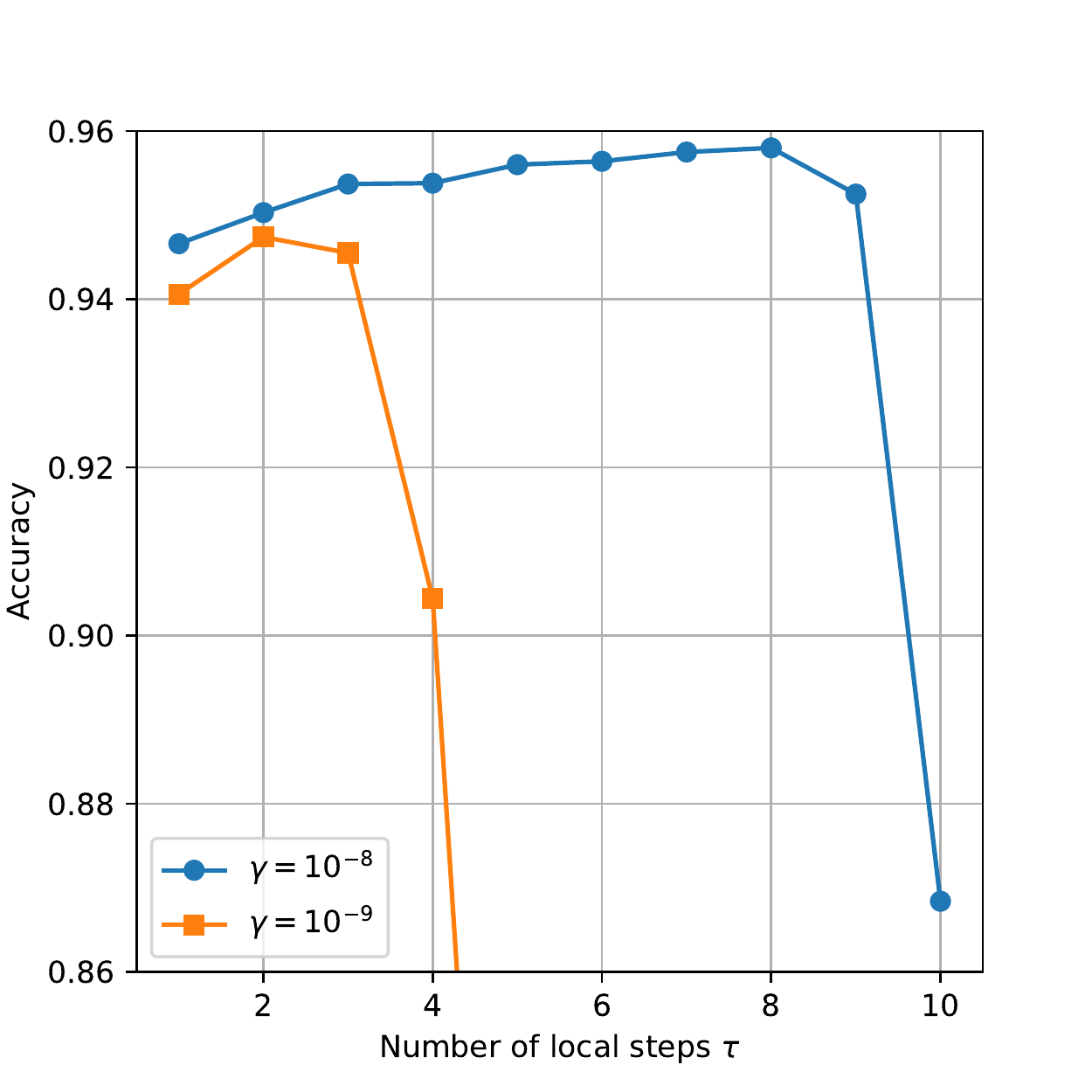}
\end{minipage}%
}%
\subfigure[CIFAR-10 dataset]{
\begin{minipage}[t]{0.5\linewidth}
\centering
\includegraphics[width=\linewidth]{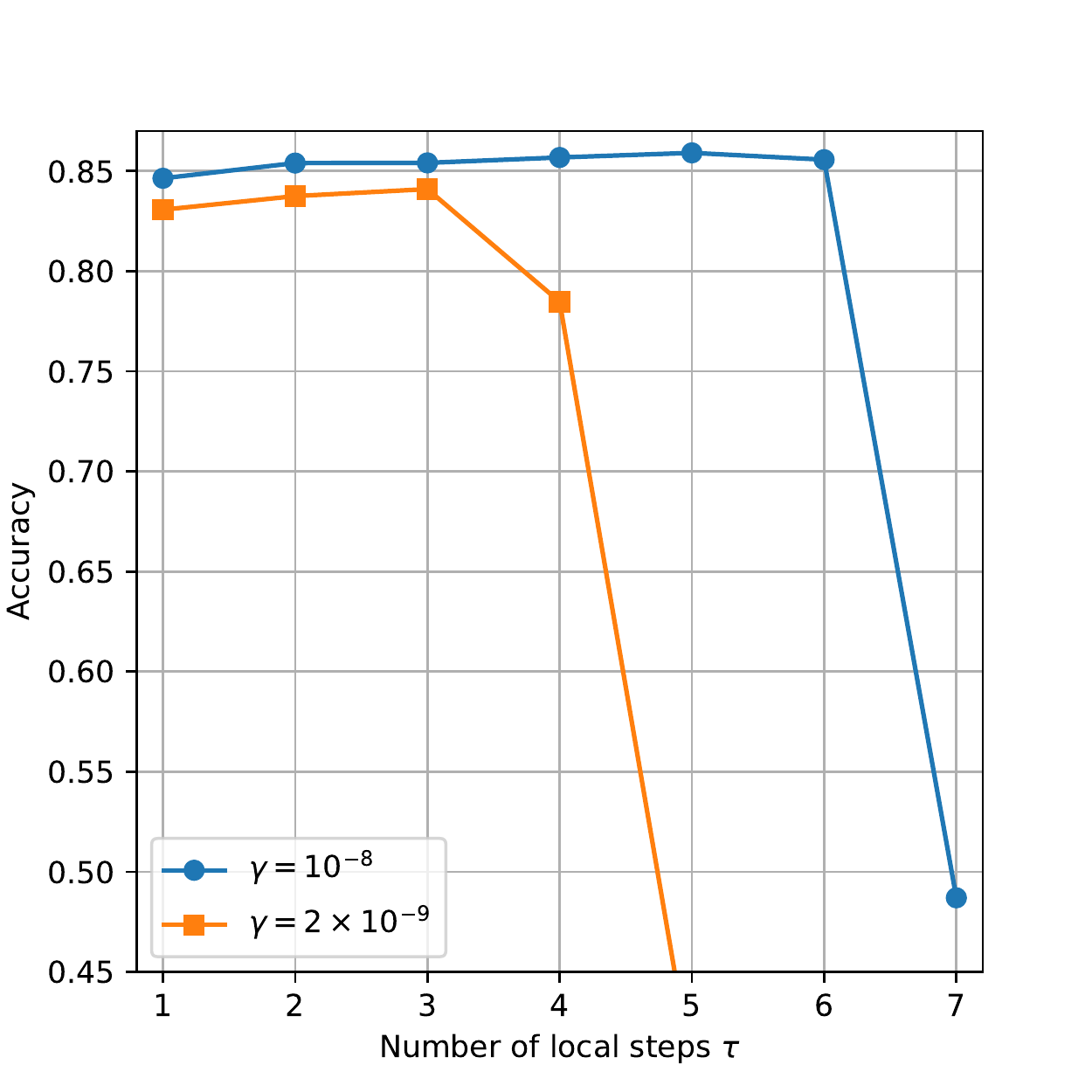}
\end{minipage}%
}%
\centering
\caption{Test accuracy with different values of $\tau$ on (a) the MNIST dataset and (b) the CIFAR-10 dataset. The model accuracy is tested after training in SCFL scheme with $T_{tot} = 20,000$ seconds.}
\label{Fig:tau}
\end{figure}

\subsubsection{Effect of Local Steps $\tau$}

In Fig. \ref{Fig:tau}, we evaluate the model accuracy of SCFL with different local training steps $\tau$ in each communication round.
In particular, we consider two cases with different average channel gains $\gamma$, where a larger value of $\gamma$ corresponds to the case with better channel quality.
In both cases, slightly increasing the local steps, e.g., from $1$ to $3$, allows the edge devices to perform more local training before gradient uploading, which improves the training efficiency by saving the number of communication rounds to achieve a target accuracy.
However, when $\tau$ further increases, e.g., $\tau > 3$, the model accuracy drops drastically due to more uplink transmission failures, and thus fewer gradients can be received at the server.
Also, with unfavorable communication quality, i.e., with a smaller value of $\gamma$, the model accuracy degrades more significantly since more edge devices struggle in completing the local training and model uploading within a communication round.
Therefore, the number of local steps should be selected properly to avoid too frequent communication or too many uplink failures.

\subsubsection{Effect of Straggler Ratio}

To simulate different levels of straggler effect, we vary the communication bandwidth to achieve a certain straggler ratio, i.e., the expected proportion of edge devices that fail to upload their local gradients to the server within a communication round.
We see from Fig. \ref{Fig:bandwidth} that all the methods suffer from accuracy drop with an increased number of stragglers.
While SCFL achieves the best test accuracy due to the use of coded data, its accuracy still decreases slightly as more stragglers cause larger gradient estimation errors in aggregation.
In comparison, FedAvg has no mechanism to handle the straggling edge devices and thus experiences more severe accuracy degradation.
Besides, CodedFedL achieves the lowest test accuracy since its efficiency is heavily degraded by frequent communications.
We notice that uploading $c=10,000$ coded data samples causes a communication overhead of $76$ MB in the MNIST training task and $156$ MB in the CIFAR-10 training task.
In comparison, the communication overhead in the training stage is $123$ MB and $208$ MB, respectively.
Despite with a non-negligible communication cost, the above results show that the coded data effectively mitigates the straggling effect and improves the learned model accuracy.

\begin{figure}[!t]
\centering
\subfigure[MNIST dataset]{
\begin{minipage}[t]{0.5\linewidth}
\centering
\includegraphics[width=\linewidth]{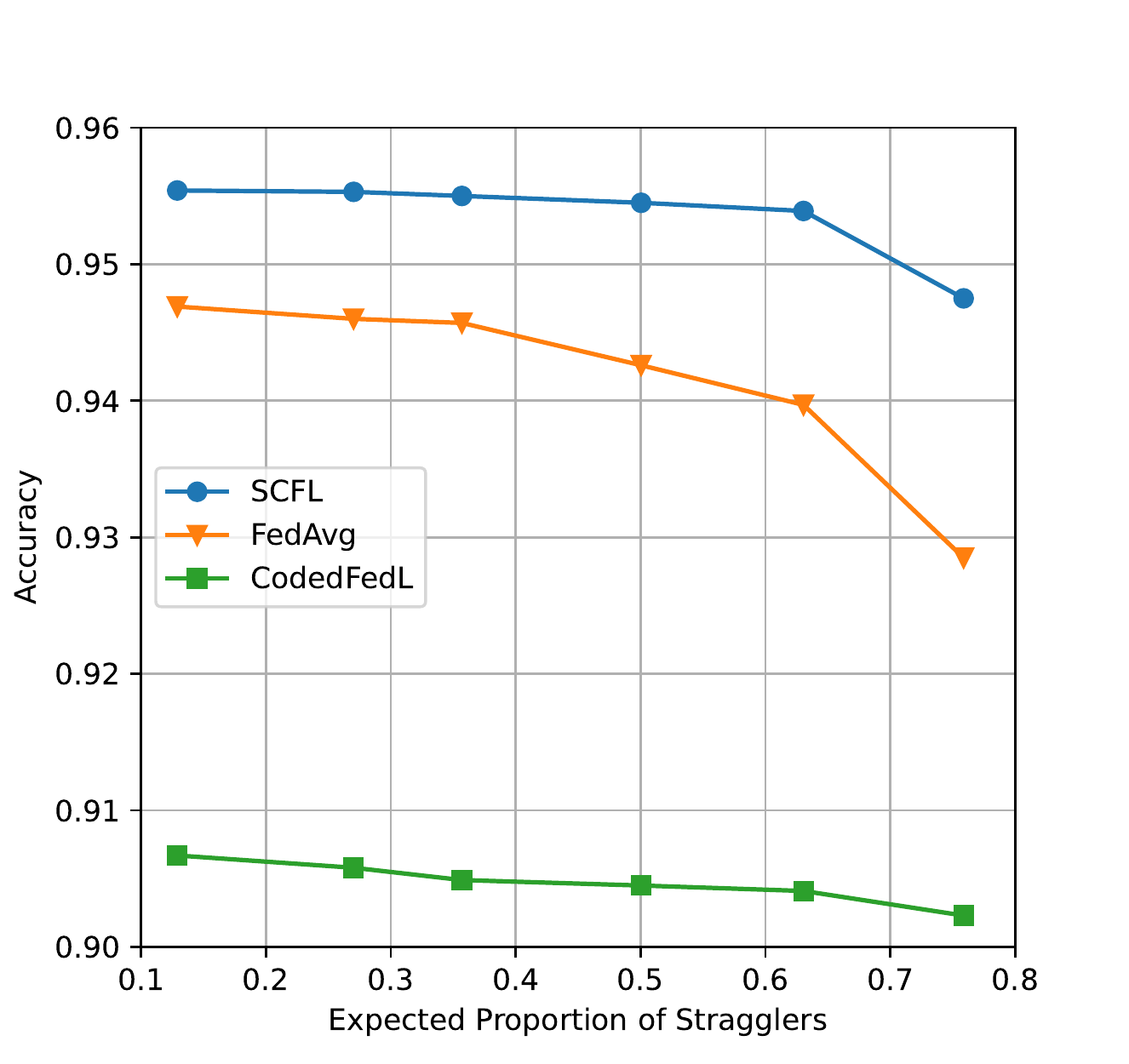}
\label{bandwidth-fig1}
\end{minipage}%
}%
\subfigure[CIFAR-10 dataset]{
\begin{minipage}[t]{0.5\linewidth}
\centering
\includegraphics[width=\linewidth]{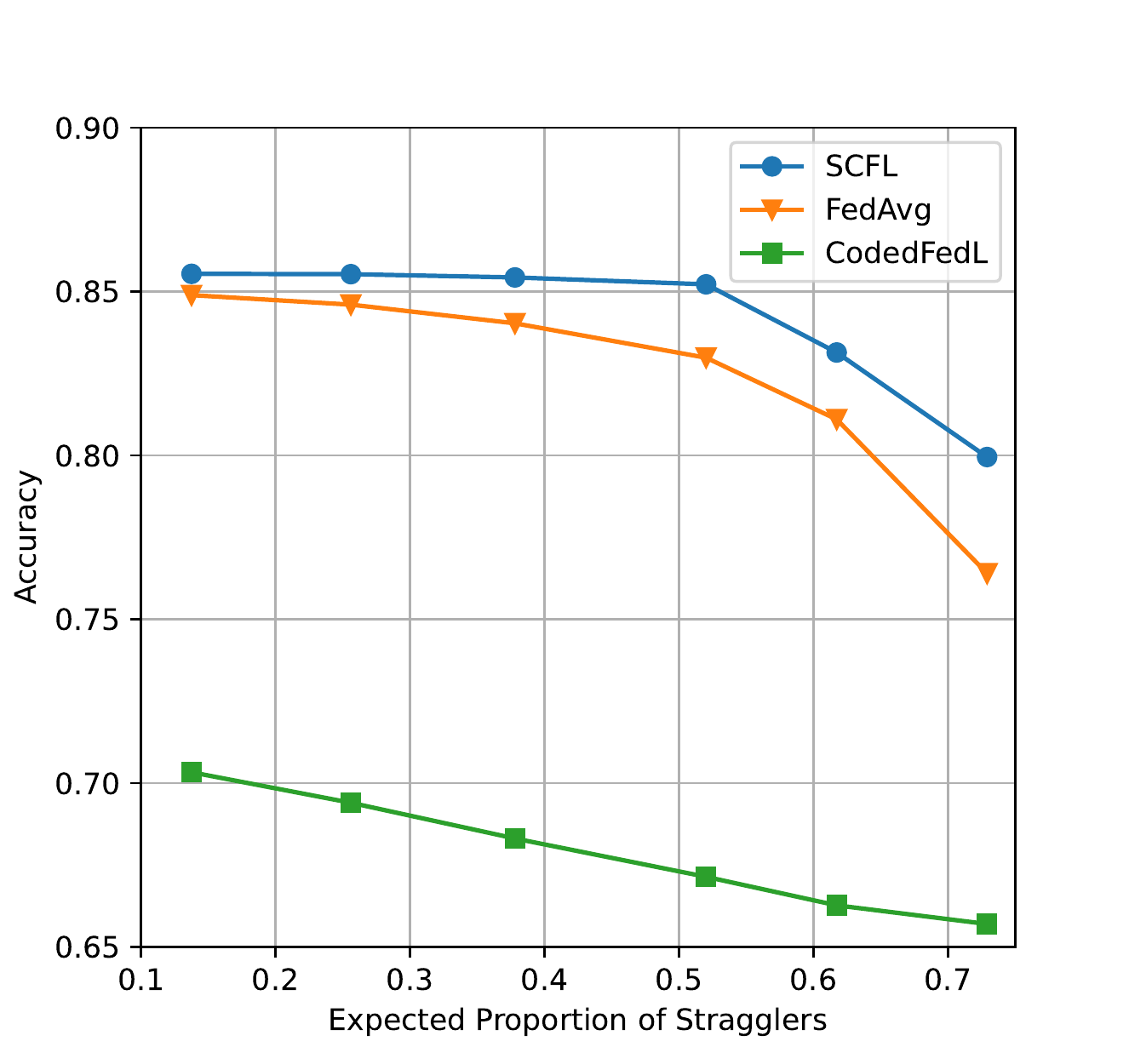}
\label{bandwidth-fig2}
\end{minipage}%
}%
\centering
\caption{Test accuracy with different expected proportions of stragglers (a) after $T_{tot} = 10,000$ seconds on the MNIST dataset and (b) after $T_{tot} = 20,000$ seconds on the CIFAR-10 dataset.}
\label{Fig:bandwidth}
\end{figure}

\subsubsection{Effect of the Number of Coded Data Samples}
We evaluate the test accuracy of SCFL on the MNIST training task with different numbers of coded data samples $c$, as shown in Fig. \ref{fig:vary-c}.
It can be observed that as the number of coded data increases, the learned model achieves a higher test accuracy, which verifies the theoretical result in Theorem \ref{thm-convergence}. Besides, a larger amount of coded data samples are needed to achieve a target accuracy when the straggler ratio becomes higher.
\begin{figure}[!t]
    \centering
    \includegraphics[width=0.35\textwidth]{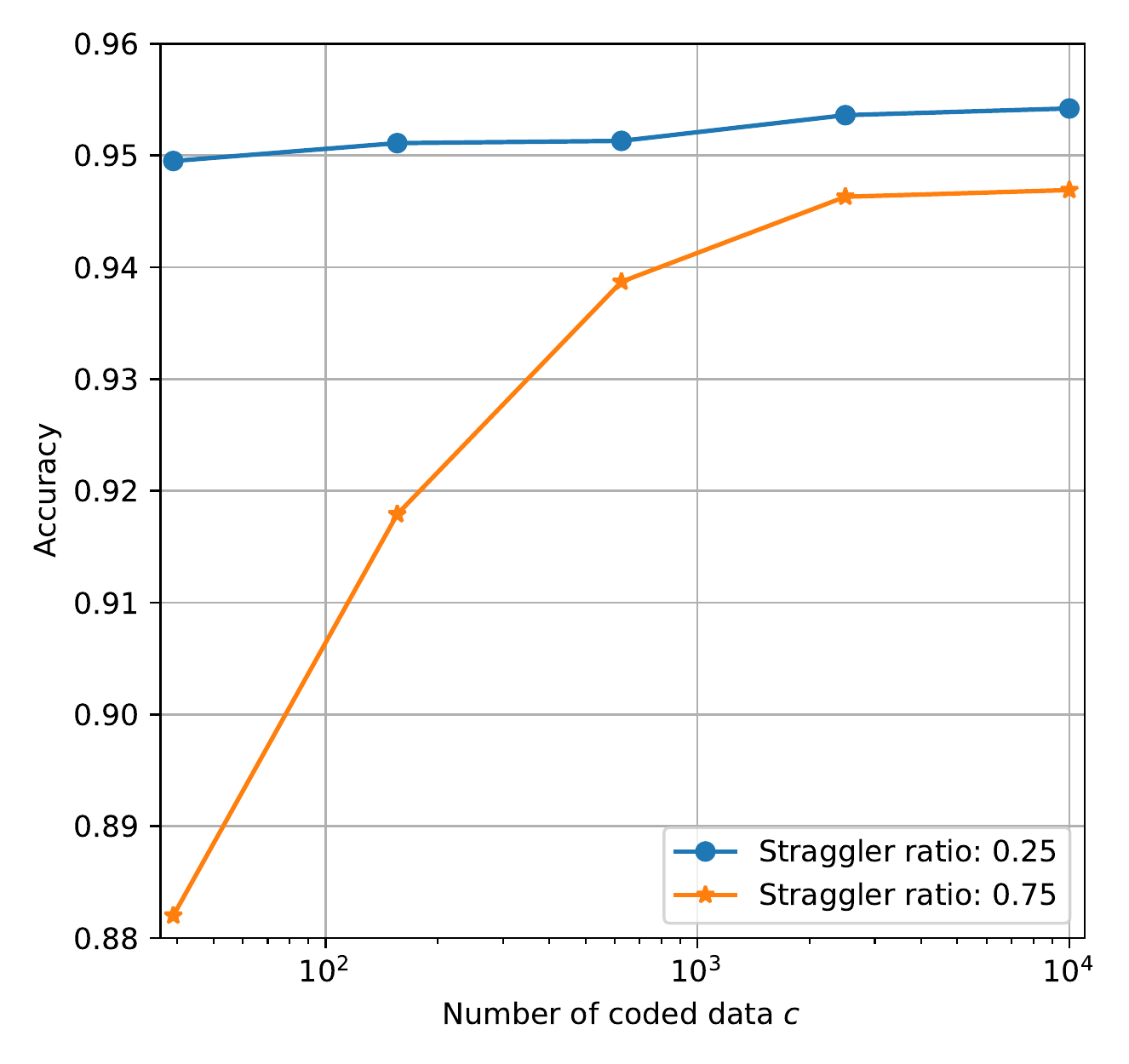}
    \caption{Test accuracy of SCFL vs. different numbers of coded data samples $c$.}
    \label{fig:vary-c}
\end{figure}

\subsubsection{Incentive Mechanism Evaluation}

\begin{figure*}[!t]
\centering
\subfigure[Device utility vs. the contract item]{
\begin{minipage}[t]{0.3\linewidth}
\centering
\includegraphics[width=\linewidth]{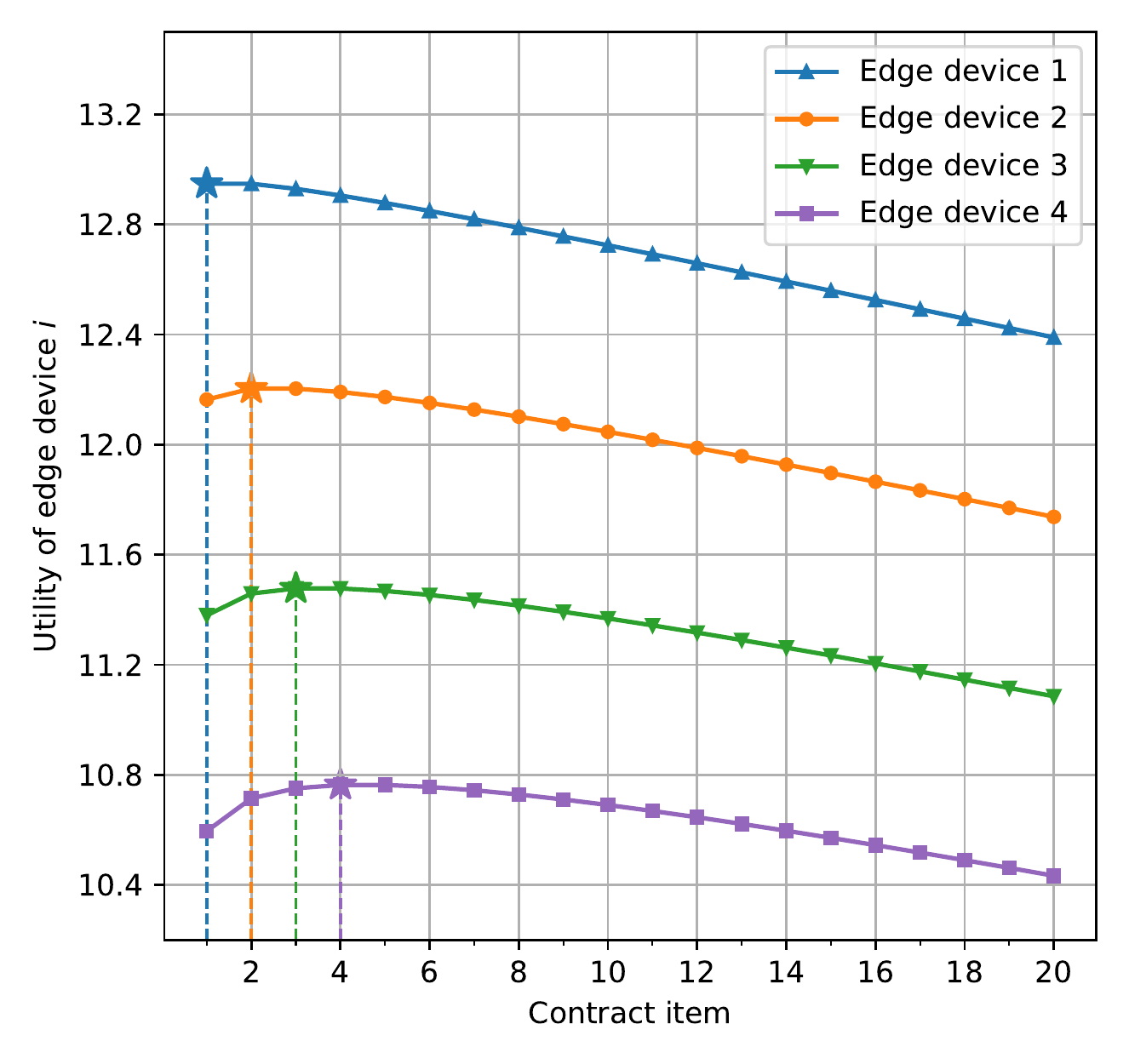}
\vspace{-3em}
\label{fig:feasibility-1}
\end{minipage}%
}%
\subfigure[Test accuracy vs. total reward]{
\begin{minipage}[t]{0.3\linewidth}
\centering
\includegraphics[width=\linewidth]{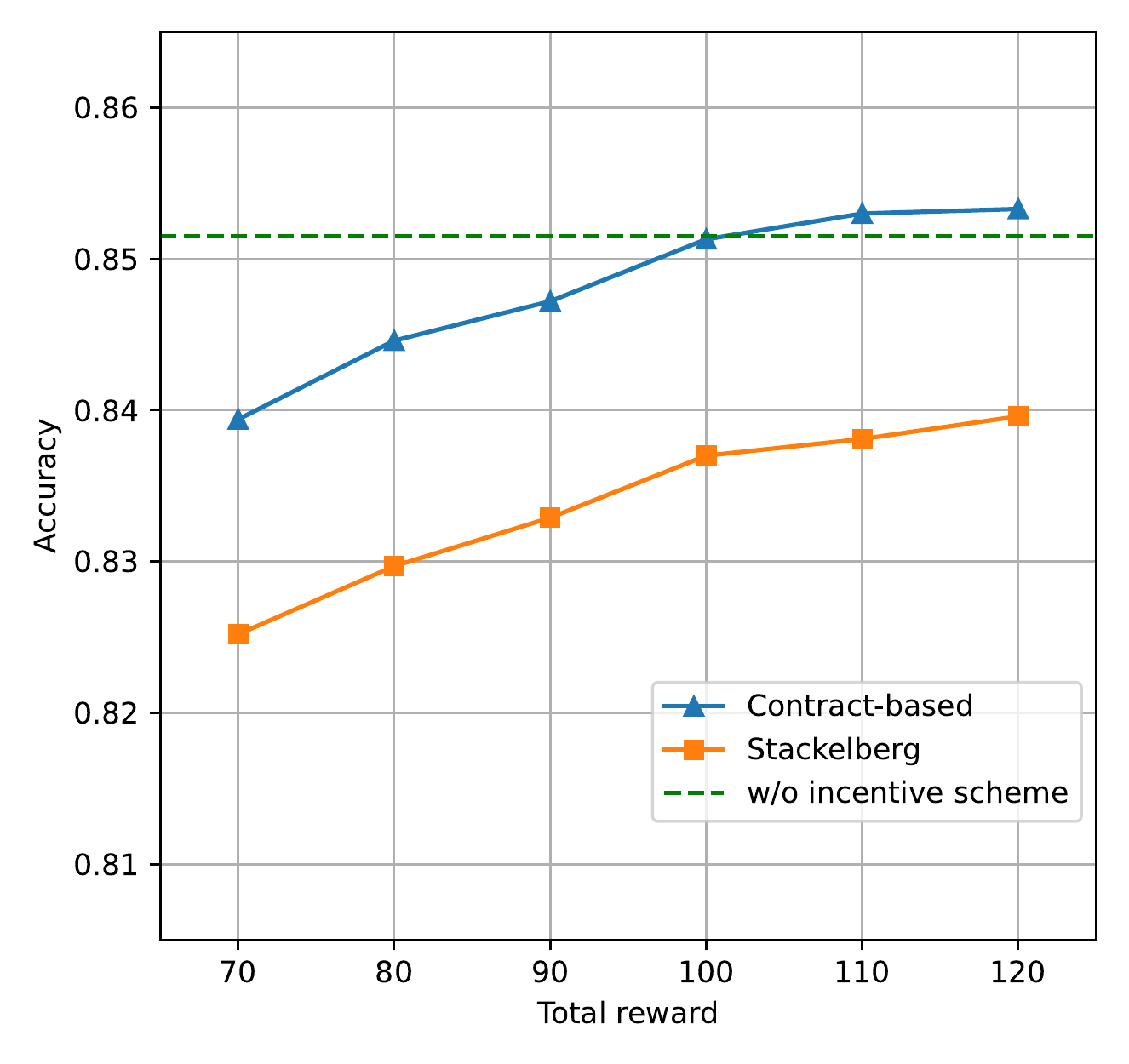}
\vspace{-3em}
\label{fig:feasibility-2}
\end{minipage}%
}
\subfigure[Test accuracy vs. total reward]{
\begin{minipage}[t]{0.3\linewidth}
\centering
\includegraphics[width=\linewidth]{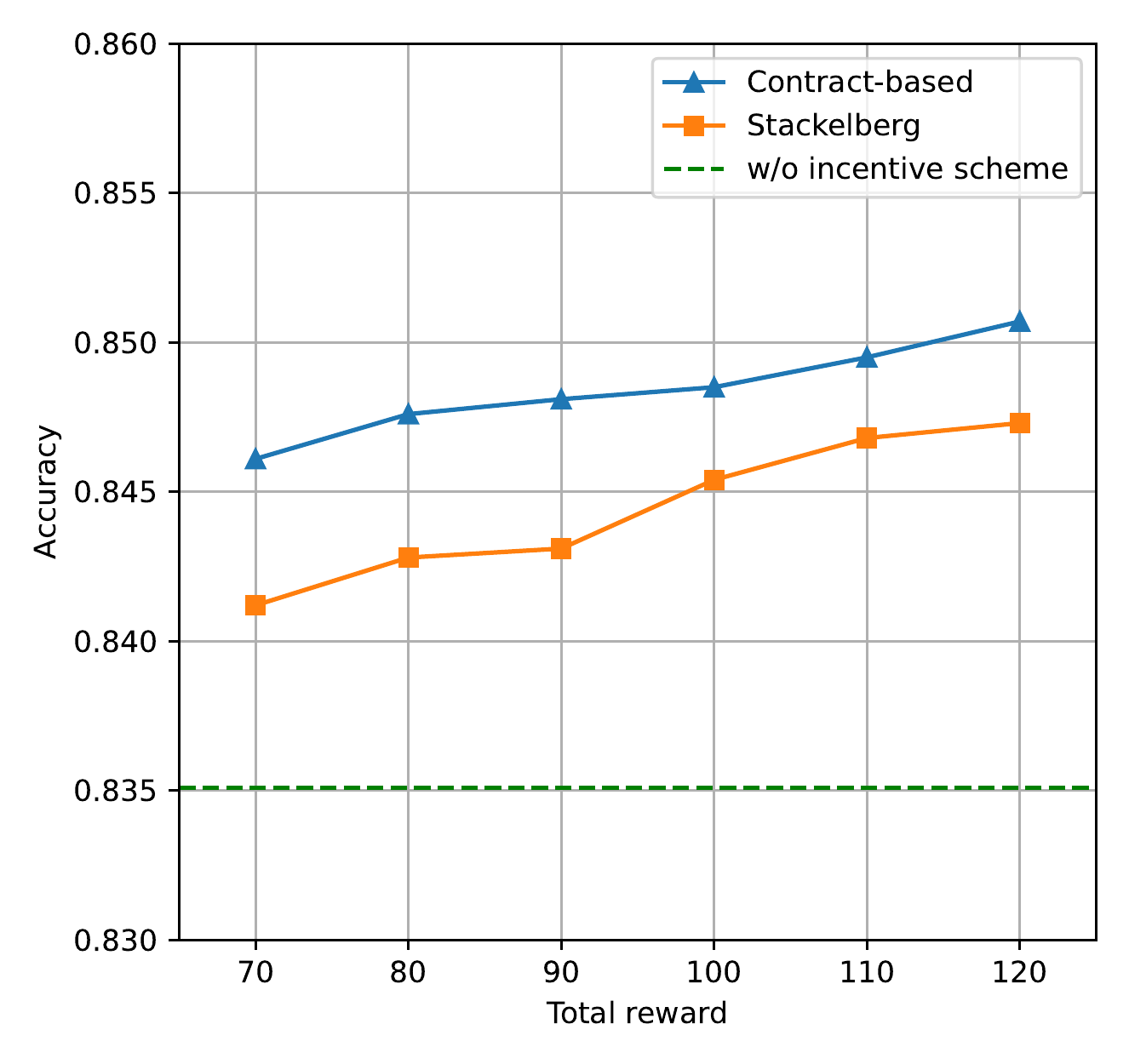}
\vspace{-3em}
\label{fig:feasibility-3}
\end{minipage}%
}
\centering
\caption{Performance comparison of incentive mechanisms on the CIFAR-10 dataset. The 20 contract items evaluated in Fig. \ref{fig:feasibility-1} satisfy $\epsilon_{1} \geq \epsilon_{2} \geq \dots \geq \epsilon_{N}$. The result in Fig. \ref{fig:feasibility-3} is evaluated with 50\% stragglers. The model accuracy is tested after training in SCFL scheme with $T_{tot} = 20,000$ seconds.}
\label{fig:feasibility}
\end{figure*}

We adopt $\Gamma(\sigma_{i}^{2}) = - \sigma_{i}^{4}$ as an example to characterize the negative effects of the added noise on the training performance.
The privacy sensitivity of each edge device is generated according to $\mu_{i} = 1 + 0.02i, \forall i\in[N]$.
In Fig. \ref{fig:feasibility-1}, we show the utility of four of the edge devices with respect to different contract items, in which, the maximal utility of edge device $i$ is achieved by accepting the $i$-th contract item instead of others, which verifies the IC requirement in \eqref{eq:IC}.
It is clear that the IR requirement in \eqref{eq:IR} is also satisfied since the maximal utility of the edge devices are all positive.

Now we validate the benefits of the proposed contract-based incentive mechanism in \cite{hu2020trading} by showing the test accuracy of the learned models with respect to the total reward paid by the server in Fig. \ref{fig:feasibility-2}-\ref{fig:feasibility-3}.
Without an incentive scheme, edge devices do not upload coded data to avoid privacy leakage, which can be viewed as FedAvg.
In the Stackelberg game approach, the server computes the Nash equilibrium to determine the total reward which is allocated to the edge devices proportionally to their privacy budgets.
To achieve a fair comparison, we vary the values of $\lambda$ in \eqref{eq:u_s_1} such that two incentive mechanisms yield the same total reward.
In particular, a small value of $\lambda$ indicates that the server concerns more on the training performance and thus provides more rewards as incentives for the edge devices, as shown in Table \ref{table:lambda} in Appendix \ref{sec:table}.
In return, the edge devices generate less noisy coded data that helps to improve the model accuracy.
By implementing an incentive mechanism, we observe severe accuracy degradation in SCFL, which is because the server receives very noisy coded data with a small reward budget (i.e., $\lambda$ is large).
However, with a sufficient amount of total reward (i.e., $\sum_{i=1}^{N} r_i \geq 100$), SCFL achieves better test accuracy than FedAvg since the server can compute gradients on the coded data to compensate the missing gradients from stragglers.
When more edge devices (e.g., 50\% in Fig. \ref{fig:feasibility-3}) become stragglers, SCFL with an incentive scheme gains larger accuracy improvement than FedAvg.
Moreover, it can be seen that under the same total reward, the contract-based mechanism leads to a better model than the Stackelberg game approach.
This is because the allocated reward to each edge device in the Stackelberg game approach depends on the privacy budgets of all edge devices.
As such, those edge devices with higher privacy sensitivity may receive a small reward and thus prefer to add much stronger noise.
On the contrary, the contract-based mechanism judiciously designs a contract item for each edge device, so that it is able to yield a less noisy coded dataset and thus improve the training performance.

\section{Conclusions}\label{sec:conclusion}

In this paper, we proposed a stochastic coded federated learning (SCFL) framework, where a coded dataset is constructed at the server to mitigate the straggler effect in FL.
We designed an unbiased aggregation scheme for SCFL which enables periodical averaging that significantly improves the training performance.
We characterized the tradeoff between the training performance and privacy guarantee of SCFL, which is determined by the noise levels of the coded datasets.
The conflict between convergence performance and the privacy protection of local coded datasets was resolved by developing a contract-based incentive mechanism.
The simulation results demonstrated the benefits of SCFL and corroborated the privacy-performance tradeoff.
We also verified the feasibility of the proposed incentive mechanism and showed it is more cost-effective than the Stackelberg game approach.

\textbf{Limitations and future works.}
While SCFL is able to solve classification problems by resorting to feature mapping methods, it is worth investigating how to extend it for non-linear ML problems in future works.
It is also important to adapt the number of coded data samples at each edge device to strike a good balance between their communication cost and privacy budget.
Furthermore, the duration of each communication round needs further optimization to achieve better training efficiency, and exploring SCFL in decentralized setups is also interesting.

\bibliographystyle{./IEEEtran}
\bibliography{ref}

\begin{IEEEbiography}[{\includegraphics[width=1in,height=1.25in,clip,keepaspectratio]{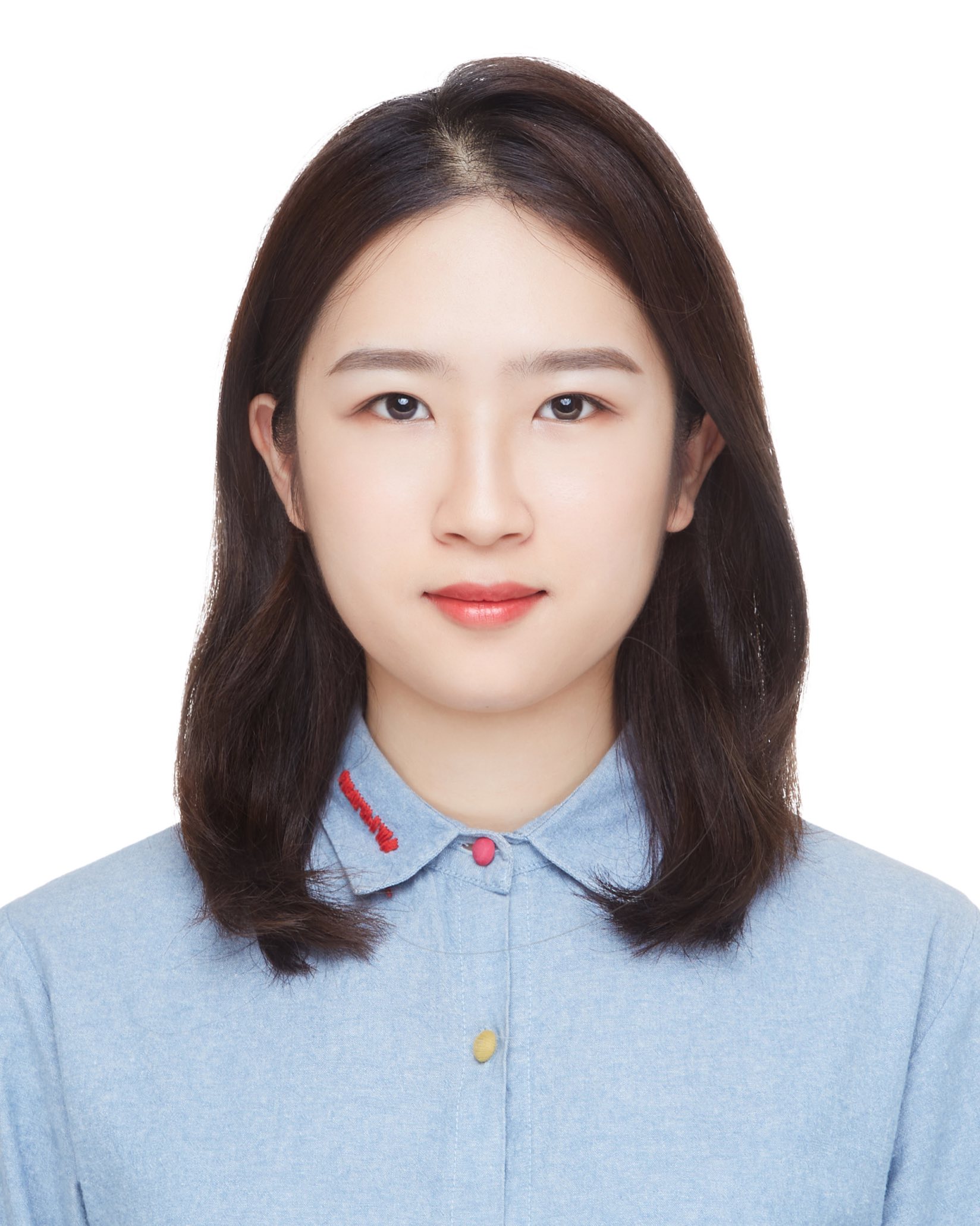}}]{Yuchang Sun}
(Graduate student member, IEEE) received the B.Eng. degree in electronic and information engineering from Beijing Institute of Technology in 2020. She is currently pursuing a Ph.D. degree at the Hong Kong University of Science and Technology. Her research interests include federated learning and distributed optimization.
\end{IEEEbiography}

\begin{IEEEbiography}[{\includegraphics[width=1in,height=1.25in,clip,keepaspectratio]{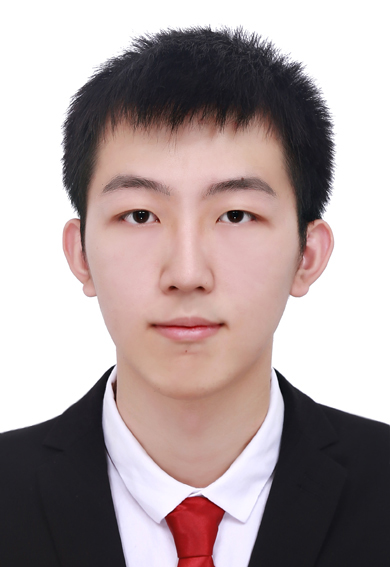}}]{Jiawei Shao}
(Graduate Student Member, IEEE) received the B.Eng. degree in telecommunication engineering from Beijing University of Posts and Telecommunications, Beijing, China, in 2019. He is currently pursuing a Ph.D. degree at the Hong Kong University of Science and Technology. His research interests include edge artificial intelligence, task-oriented communications, neural compression, and federated learning.
\end{IEEEbiography}

\begin{IEEEbiography}[{\includegraphics[width=1in,height=1.25in,clip,keepaspectratio]{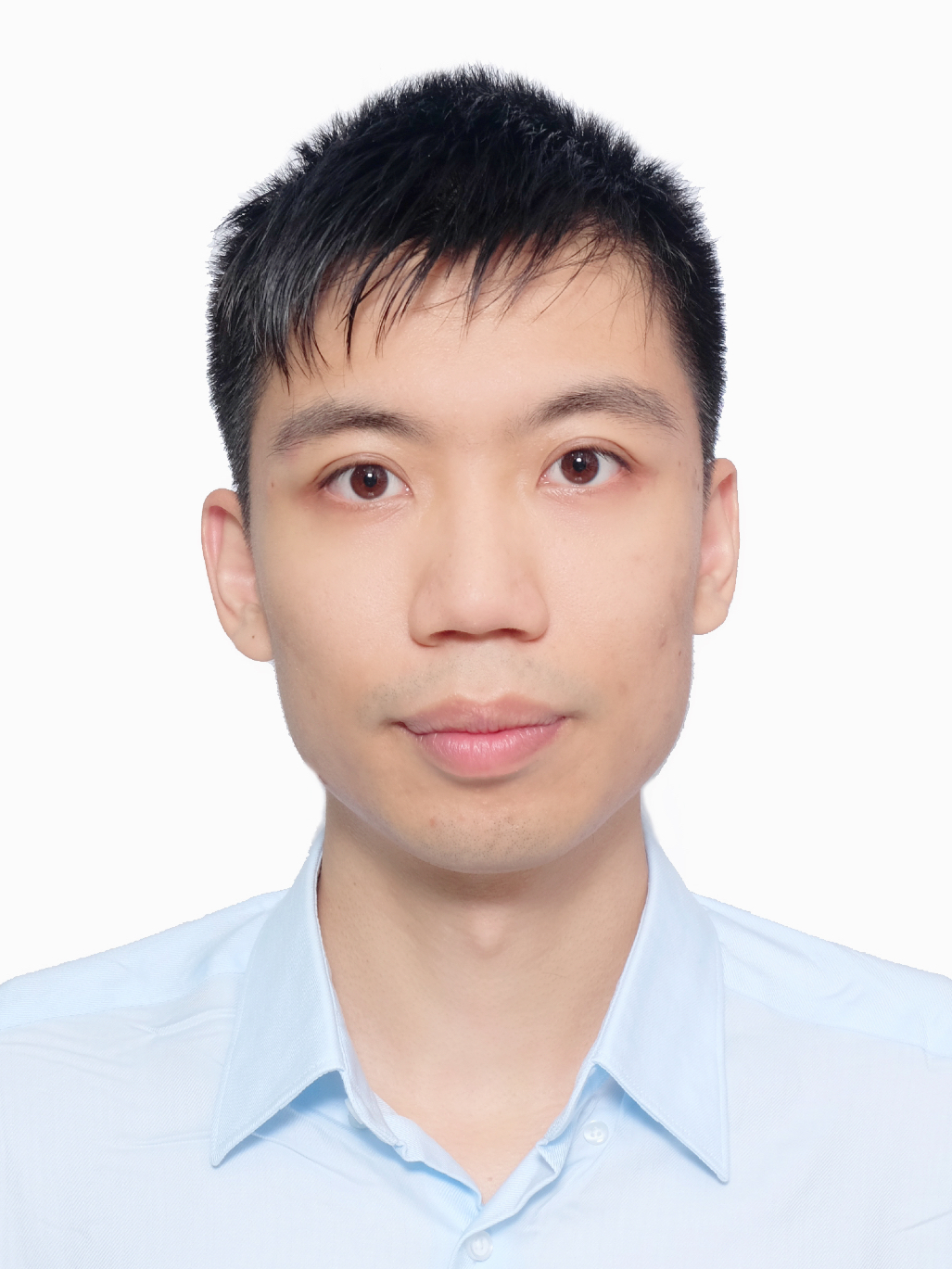}}]{Yuyi Mao}
(Member, IEEE) received the B.Eng. degree in information and communication engineering from Zhejiang University, Hangzhou, China, in 2013, and the Ph.D. degree in electronic and computer engineering from The Hong Kong University of Science and Technology, Hong Kong, in 2017. He was a Lead Engineer with the Hong Kong Applied Science and Technology Research Institute Co., Ltd., Hong Kong, and a Senior Researcher with the Theory Lab, 2012 Labs, Huawei Tech. Investment Co., Ltd., Hong Kong. He is currently a Research Assistant Professor with the Department of Electrical and Electronic Engineering, The Hong Kong Polytechnic University, Hong Kong. His research interests include wireless communications and networking, mobile-edge computing and learning, and wireless artificial intelligence.

He was the recipient of the 2021 IEEE Communications Society Best Survey Paper Award and the 2019 IEEE Communications Society and Information Theory Society Joint Paper Award. He was also recognized as an Exemplary Reviewer of the IEEE Wireless Communications Letters in 2021 and 2019 and the IEEE Transactions on Communications in 2020. He is an Editor of the IEEE Wireless Communications Letters and an Associate Editor of the EURASIP Journal on Wireless Communications and Networking.
\end{IEEEbiography}

\begin{IEEEbiography}[{\includegraphics[width=1in,height=1.25in,clip,keepaspectratio]{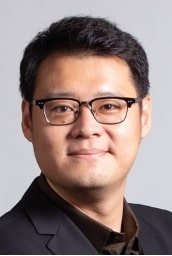}}]{Songze Li}
(Member, IEEE) received the B.Sc. degree from New York University in 2011 and the Ph.D. degree from University of Southern California in 2018, both in electrical engineering. He is a professor at School of Cyber Science and Engineering, Southeast University, China. Dr. Li’s research interest is on developing secure, scalable, and accurate distributed computing and learning solutions, mainly focused on the areas of coded distributed computing, federated learning, and blockchains. He was Qualcomm Innovation Fellowship finalist in 2017. He received the Best Paper Award at NeurIPS-20 Workshop on Scalability, Privacy, and Security in Federated Learning.
\end{IEEEbiography}

\begin{IEEEbiography}[{\includegraphics[width=1in,height=1.25in,clip,keepaspectratio]{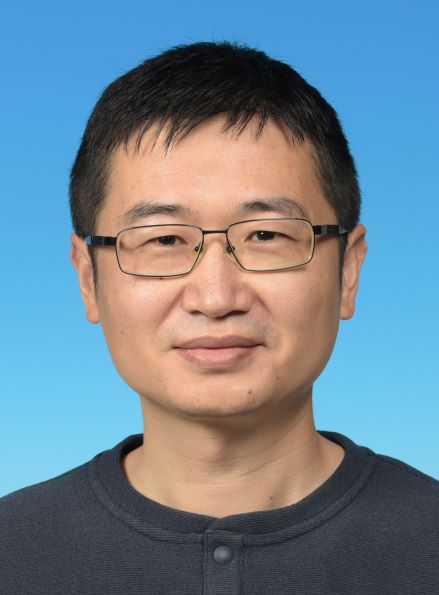}}]{Jun Zhang}
(Fellow, IEEE) received the B.Eng. degree in Electronic Engineering from the University of Science and Technology of China in 2004, the M.Phil. degree in Information Engineering from the Chinese University of Hong Kong in 2006, and the Ph.D. degree in Electrical and Computer Engineering from the University of Texas at Austin in 2009. He is an Associate Professor in the Department of Electronic and Computer Engineering at the Hong Kong University of Science and Technology. His research interests include wireless communications and networking, mobile edge computing and edge AI, and cooperative AI.

Dr. Zhang co-authored the book Fundamentals of LTE (Prentice-Hall, 2010). He is a co-recipient of several best paper awards, including the 2021 Best Survey Paper Award of the IEEE Communications Society, the 2019 IEEE Communications Society \& Information Theory Society Joint Paper Award, and the 2016 Marconi Prize Paper Award in Wireless Communications. Two papers he co-authored received the Young Author Best Paper Award of the IEEE Signal Processing Society in 2016 and 2018, respectively. He also received the 2016 IEEE ComSoc Asia-Pacific Best Young Researcher Award. He is an Editor of IEEE Transactions on Communications, IEEE Transactions on Machine Learning in Communications and Networking, and was an editor of IEEE Transactions on Wireless Communications (2015-2020). He served as a MAC track co-chair for IEEE Wireless Communications and Networking Conference (WCNC) 2011 and a co-chair for the Wireless Communications Symposium of IEEE International Conference on Communications (ICC) 2021. He is an IEEE Fellow and an IEEE ComSoc Distinguished Lecturer.
\end{IEEEbiography}

\appendices

\section{Proof of Lemma \ref{lem:unbiased}}\label{proof:unbiased}

We begin with an important lemma that characterizes the relationships between the stochastic gradients in the $u$-th step of the $k$-th communication round.
\begin{lemma}\label{lem:equal}
For any $u\in[\tau]$, $k\in[K]$, $i\in[N]$, the following equality holds:
\begin{equation}
\setlength\abovedisplayskip{0cm}
\setlength\belowdisplayskip{-1cm}
    \mathbb{E}[ g_{k,u}^{(i)}( \mathbf{W}_{k,u}^{(i)} ) ] \!=\! \mathbf{Z}^{(i)} \nabla f(\mathbf{W}_{k,u}), \mathbb{E}[g_{k,u}^{\s}( \mathbf{W}_{k,u}^{\s} )] = \nabla f(\mathbf{W}_{k,u}).
    \label{eq:equal}
\end{equation}
\end{lemma}
\begin{proof}
We prove Lemma \ref{lem:equal} by mathematical induction.
Recall $\nabla f(\mathbf{W}_{k,u}) = \mathbf{X}^{\TT} \left(\mathbf{X} \mathbf{W}_{k,u} - \mathbf{Y}\right)$.
It is straightforward that \eqref{eq:equal} holds for $u=0$ since the models have the same initialization, i.e., $\mathbf{W}_{k,0}^{(i)} = \mathbf{W}_{k,0}^{\s} = \mathbf{W}_{k,0} = \mathbf{W}_{k}$.
Suppose $\mathbb{E}[\mathbf{W}_{k,u}^{(i)}] = \mathbb{E}[\mathbf{W}_{k,u}^{\s}] = \mathbb{E}[\mathbf{W}_{k,u}], \forall u=0,1,\dots, \tau-1$, and thus we have:
\begin{align}
\setlength\abovedisplayskip{0cm}
\setlength\belowdisplayskip{-1em}
    & \mathbb{E} \left[g_{k,u}^{(i)}(\mathbf{W}_{k,u}^{(i)}) \right] \nonumber \\
    =& \mathbb{E} \left[ \mathbf{X}^{(i)\TT} \left( \frac{l_i}{b_k^i} \mathbf{S}_{k,u}^{(i)\TT} \mathbf{S}_{k,u}^{(i)} \right) \left(\mathbf{X}^{(i)} \mathbf{W}_{k,u}^{(i)} - \mathbf{Y}^{(i)} \right) \right] \\
    \overset{\text{(a)}}{=} &\mathbb{E} \left[ \mathbf{X}^{(i)\TT} \left(\mathbf{X}^{(i)} \mathbf{W}_{k,u}^{(i)} - \mathbf{Y}^{(i)} \right)\right] \nonumber \\
    =& \left( \mathbf{Z}^{(i)\TT} \mathbf{Z}^{(i)} \right) \mathbf{X}^{\TT} \left(\mathbf{X} \mathbf{W}_{k,u}^{(i)} - \mathbf{Y}\right) \nonumber \\
    =& \mathbf{Z}^{(i)} \mathbf{X}^{\TT} \left(\mathbf{X} \mathbf{W}_{k,u} - \mathbf{Y}\right), \nonumber
\end{align}
and
\begin{align}
    & \mathbb{E} \left[ g_{k,u}^{\s}(\mathbf{W}_{k,u}^{\s}) \right] \nonumber \\
    =& \mathbb{E} \left[ \mathbf{\tilde{X}}^{\TT} \left( \frac{1}{b_{\s}} (\mathbf{S}_{k}^{\s})^{\TT} \mathbf{S}_{k,u}^{\s} \right) \left( \mathbf{\tilde{X}} \mathbf{W}_{k,u}^{\s} - \mathbf{\tilde{Y}} \right) - \sigma^2 \mathbf{W}_{k,u}^{\s} \right] \\
    \overset{\text{(b)}}{=} & \mathbb{E} \left[  \mathbf{\tilde{X}}^{\TT} \left( \mathbf{\tilde{X}} \mathbf{W}_{k,u}^{\s} - \mathbf{\tilde{Y}} \right) - \sigma^2 \mathbf{W}_{k,u}^{\s} \right] \nonumber\\
    = & \mathbb{E} \left[ (\mathbf{X})^{\TT} \left( \frac{1}{c} \mathbf{G}^{\TT} \mathbf{G} \right) \mathbf{X} \mathbf{W}_{k,u}^{\s} \right]
    + \mathbb{E} \left[ \mathbf{X}^{\TT} \left( \frac{\sigma^2}{c} \mathbf{G}^{\TT} \mathbf{N} \right) \mathbf{X} \mathbf{W}_{k,u}^{\s} \right] \nonumber \\
    &- \mathbb{E} \left[ \mathbf{X}^{\TT} \left( \frac{1}{c} \mathbf{G}^{\TT} \mathbf{G} \right) \mathbf{Y} \right] 
    + \mathbb{E} \left[ \left( \frac{\sigma^2}{c} \mathbf{N}^{\TT} \mathbf{G} \right) \mathbf{X} \mathbf{W} \right] \nonumber \\
    &+ \mathbb{E} \left[ \frac{\sigma^2}{c} \mathbf{N}^{\TT} \mathbf{N} \mathbf{W}_{k,u}^{\s} \right]
    + \mathbb{E} \left[ \left( \frac{\sigma^2}{c} \mathbf{N}^{\TT} \mathbf{G} \right) \mathbf{Y} \right] 
    - \sigma^2 \mathbf{W}_{k,u}^{\s} \nonumber \\
    \overset{\text{(c)}}{=} & \mathbf{X}^{\TT} \left(\mathbf{X} \mathbf{W}_{k,u} - \mathbf{Y}\right), \nonumber
\end{align}
where (a), (b), and (c) adopt the properties in Lemma \ref{useful-lemma}.
Then for $u+1$ we have $\mathbb{E}[ \mathbf{W}_{k,u+1}] = \mathbb{E}[\mathbf{W}_{k,u}] - \mathbb{E} \left[ \frac{\eta_k}{2} \left( \sum_{i=1}^{N} \frac{1}{p_i} g_{k,u}^{(i)}( \mathbf{W}_{k,u}^{(i)} ) \mathbbm{1}_{k}^{(i)} + g_{k,u}^{\s}( \mathbf{W}_{k,u}^{\s} ) \right)  \right ] = \mathbb{E}[\mathbf{W}_{k,u} - \nabla f(\mathbf{W}_{k,u})] = \mathbb{E}[\mathbf{W}_{k,u+1}]$.
Therefore, Lemma \ref{lem:equal} holds.
\end{proof}

Since $\tau$ steps of SGD are independent, we can decompose the aggregated gradient $g(\mathbf{W}_{k})$ into $\tau$ parts, i.e., $g(\mathbf{W}_{k}) = \sum_{u=0}^{\tau-1} \frac{1}{2} \left( \sum_{i=1}^{N} \frac{1}{p_i} g_{k,u}^{(i)}(\mathbf{W}_{k,u}^{(i)}) \mathbbm{1}_{k}^{(i)} + g_{k,u}^{\s}(\mathbf{W}_{k,u}^{\s}) - \sigma^2 \mathbf{W}_{k,u}^{\s} \right)$.
According to Lemma \ref{lem:equal}, we show the gradient obtained in step $u$ is an unbiased estimate of $\nabla f(\mathbf{W}_{k,u})$ as follows:
\begin{small}
\begin{align}
    & \mathbb{E} \left[\frac{1}{2} \left( \sum_{i=1}^{N} \frac{1}{p_i} g_{k,u}^{(i)}(\mathbf{W}_{k,u}^{(i)}) \mathbbm{1}_{k}^{(i)} + g_{k,u}^{\s}(\mathbf{W}_{k,u}^{\s}) \right) \right] \label{eq:28}\\
    =& \frac{1}{2} \mathbb{E} \left[ \sum_{i=1}^{N}  \mathbf{X}^{(i)\TT} \left( \frac{l_i}{b_k^i} \mathbf{S}_{k,u}^{(i)\TT} \mathbf{S}_{k,u}^{(i)} \right) \left(\mathbf{X}^{(i)} \mathbf{W}_{k,u}^{(i)} - \mathbf{Y}^{(i)} \right) \right. \nonumber \\
    &+ \left. \frac{1}{c} \mathbf{\tilde{X}}^{\TT} \left( \frac{c}{b_{\s}} (\mathbf{S}_{k,u}^{\s})^{\TT} \mathbf{S}_{k,u}^{\s} \right) \left( \mathbf{\tilde{X}} \mathbf{W}_{k,u}^{\s} - \mathbf{\tilde{Y}} \right) - \sigma^2 \mathbf{W}_{k,u}^{\s} \right] \nonumber \\
    \overset{\text{(d)}}{=}& \frac{1}{2} \mathbb{E} \left[ \sum_{i=1}^{N}  \mathbf{X}^{(i)\TT} \left(\mathbf{X}^{(i)} \mathbf{W}_{k,u}^{(i)} - \mathbf{Y}^{(i)} \right) + \frac{1}{c} \mathbf{\tilde{X}}^{\TT} \left( \mathbf{\tilde{X}} \mathbf{W}_{k,u}^{\s} - \mathbf{\tilde{Y}} \right) - \sigma^2 \mathbf{W}_{k,u}^{\s} \right] \nonumber \\
    =& \frac{1}{2} \sum_{i=1}^{N}  \mathbf{X}^{(i)\TT} \left(\mathbf{X}^{(i)} \mathbf{W}_{k,u}^{(i)} - \mathbf{Y}^{(i)} \right)
    + \frac{1}{2} \mathbb{E} \left[ (\mathbf{X})^{\TT} \left( \frac{1}{c} \mathbf{G}^{\TT} \mathbf{G} \right) \mathbf{X} \mathbf{W}_{k,u}^{(i)} \right] \nonumber \\
    &+ \frac{1}{2} \mathbb{E} \left[ (\mathbf{X})^{\TT} \left( \frac{\sigma^2}{c} \mathbf{G}^{\TT} \mathbf{N} \right) \mathbf{X} \mathbf{W}_{k,u}^{(i)} \right] 
    + \frac{1}{2} \mathbb{E} \left[ (\mathbf{X})^{\TT} \left( \frac{1}{c} \mathbf{G}^{\TT} \mathbf{G} \right) \mathbf{Y} \right] \nonumber \\
    &+ \frac{1}{2} \mathbb{E} \left[ \frac{\sigma^2}{c} \mathbf{N}^{\TT} \mathbf{N} \right] 
    + \frac{1}{2} \mathbb{E} \left[ \left( \frac{\sigma^2}{c} \mathbf{N}^{\TT} \mathbf{G} \right) \mathbf{Y} \right]
    - \frac{1}{2} \sigma^2 \mathbf{W}_{k,u}^{\s} \nonumber \\
    \overset{\text{(e)}}{=} & \frac{1}{2} \left[ \sum_{i=1}^{N} \left( \mathbf{Z}^{(i)\TT} \mathbf{Z}^{(i)} \right) \mathbf{X}^{\TT} \left(\mathbf{X} \mathbf{W}_{k,u}^{(i)} - \mathbf{Y}\right) + \mathbf{X}^{\TT} \left(\mathbf{X} \mathbf{W}_{k,u} - \mathbf{Y}\right) \right] \nonumber \\
    =& \mathbf{X}^{\TT} \left(\mathbf{X} \mathbf{W}_{k,u} - \mathbf{Y}\right), \nonumber
\end{align}
\end{small}where (d) and (e) follow the properties derived in Lemma \ref{useful-lemma}.
Since the variance in each step is independent, the result in Lemma \ref{lem:unbiased} is completed by summing up both sides of \eqref{eq:28} over $u=0, 1, \dots, \tau-1$.

\section{Proof of Lemma \ref{lem:device}}\label{appendix-B}

We decompose the gradient variance into each SGD step $u$ as follows and provide an upper bound for them respectively. 
\begin{align}
    & \mathbb{E} \Bigg[\Bigg\| \sum_{u=0}^{\tau-1}  \sum_{i=1}^{N} \frac{1}{p_i} g_{k,u}^{(i)}(\mathbf{W}_{k,u}^{(i)}) \mathbbm{1}_{k}^{(i)} - \nabla f(\mathbf{W}_{k,u}) \Bigg\|_{\mathrm{F}}^2 \Bigg] \nonumber \\
    \leq & \tau \sum_{u=0}^{\tau-1} \mathbb{E} \left[\left\| \sum_{i=1}^{N} \frac{1}{p_i} g_{k,u}^{(i)}(\mathbf{W}_{k,u}^{(i)}) \mathbbm{1}_{k}^{(i)} - \nabla f(\mathbf{W}_{k,u}) \right\|_{\mathrm{F}}^2 \right].
\label{eq:help-0}
\end{align}

Define the full-batch gradient on edge device $i$ as $\nabla f_i (\mathbf{W}_{k,u}^{(i)})  \triangleq \mathbf{X}^{(i)\TT} (\mathbf{X}^{(i)} \mathbf{W}_{k,u}^{(i)} - \mathbf{Y}^{(i)})$.
For any step $u$, we have:
\begin{small}
\begin{align}
    & \mathbb{E} \Bigg[\Bigg\| \sum_{i=1}^{N} \frac{1}{p_i} g_{k,u}^{(i)}(\mathbf{W}_{k,u}^{(i)}) \mathbbm{1}_{k}^{(i)} - \nabla f(\mathbf{W}_{k,u}) \Bigg\|_{\mathrm{F}}^2 \Bigg] \nonumber \\
    \overset{\text{(a)}}{\leq} & 2 \mathbb{E} \Bigg[ \Bigg\| \sum_{i=1}^{N} \frac{ \mathbbm{1}_{k}^{(i)} - p_i}{p_i} g_{k,u}^{(i)}(\mathbf{W}_{k,u}^{(i)}) \Bigg\|_{\mathrm{F}}^2 \Bigg] \nonumber \\
    &+2 \mathbb{E} \Bigg[\Bigg\| \sum_{i=1}^{N} g_i (\mathbf{W}_{k,u}^{(i)}) - \nabla f_i (\mathbf{W}_{k,u}^{(i)}) \Bigg\|_{\mathrm{F}}^2 \Bigg] \nonumber \\
    \overset{\text{(b)}}{=} & 2 \mathbb{E} \left[ \sum_{i=1}^{N} \frac{ \mathbb{E} [\| \mathbbm{1}_{k}^{(i)} - p_i \|^2]}{(p_i)^2} \left\| g_{k,u}^{(i)}(\mathbf{W}_{k,u}^{(i)}) \right\|_{\mathrm{F}}^2 \right] \nonumber \\
    &+ 2 \mathbb{E} \Bigg[\Bigg\| \mathbf{X}^{(i)\TT} \Bigg( \sum_{i=1}^{N} \frac{l_i}{b_k^i} \mathbf{S}_{k,u}^{(i)}-\mathbf{I} \Bigg)  (\mathbf{X}^{(i)} \mathbf{W}_{k,u}^{(i)} - \mathbf{Y}^{(i)}) \Bigg\|_{\mathrm{F}}^2 \Bigg] \nonumber\\
    \overset{\text{(c)}}{\leq} & 2 \sum_{i=1}^{N} \frac{1 - p_i}{p_i} \zeta_i^2\kappa_i^2 + 
    2 \sum_{i=1}^{N} \frac{l_i(l_i-b_k^i)}{b_k^i} \zeta_i^2\kappa_i^2,
\label{eq:help-1}
\end{align}
\end{small}
where (a) follows the Jensen's inequality, and (b) holds since $\sum_{i=1}^{N} {\mathbf{Z}^{(i)\TT}}\mathbf{Z}^{(i)} = \mathbf{I}_m$ and $\sum_{i=1}^{N}$ $\nabla f_i (\mathbf{W}_{k,u}^{(i)}) = \nabla f (\mathbf{W}_{k,u}^{(i)})$.
Besides, the first term in (c) follows $\mathbb{E} [\|\mathbbm{1}_{k}^{(i)} - p_i\|^2] = p_i (1 - p_i)$, while the second term adopts the result in Lemma \ref{useful-lemma}.
By combing \eqref{eq:help-0} and \eqref{eq:help-1}, we conclude the proof.

\section{Proof of Lemma \ref{lem:server}}\label{appendix-C}
Similar to the proof of Lemma \ref{lem:device}, we first decompose the gradient variance into $\tau$ steps:
\begin{align}
    & \mathbb{E} \Bigg[\Bigg\| \sum_{u=0}^{\tau-1}  g_{k,u}^{\s}(\mathbf{W}_{k,u}^{\s}) - \sigma^2 \mathbf{W}_{k,u}^{\s} - \nabla f(\mathbf{W}_{k,u}) \Bigg\|_{\mathrm{F}}^2 \Bigg] \nonumber \\
    \leq & \tau \sum_{u=0}^{\tau-1} \mathbb{E} \left[\left\| g_{k,u}^{\s}(\mathbf{W}_{k,u}^{\s}) - \sigma^2 \mathbf{W}_{k,u}^{\s} - \nabla f(\mathbf{W}_{k,u}) \right\|_{\mathrm{F}}^2 \right].
\end{align}
Define the full-batch gradient on the coded dataset $(\mathbf{\tilde{X}},\mathbf{\tilde{Y}})$ as $\nabla f_\mathrm{s} (\mathbf{W}_{k,u}^{\s})  \triangleq  \frac{1}{c} \mathbf{\tilde{X}}^{\TT} (\mathbf{\tilde{X}} \mathbf{W}_{k,u}^{\s} - \mathbf{\tilde{Y}})$. 
Given the independent error caused by mini-batch sampling in \eqref{eq:help-2} and the data coding in \eqref{eq:help-3}, we derive the gradient estimation error over the coded data as follows:
\begin{align}
    & \mathbb{E} \left[ \left\| g_{k,u}^{\s}(\mathbf{W}_{k,u}^{\s}) -\sigma^2 \mathbf{W}_{k,u}^{\s} - \nabla f(\mathbf{W}_{k,u}^{\s}) \right\|_{\mathrm{F}}^2 \right] \nonumber\\
    = & \mathbb{E} \left[\left\| g_{k,u}^{\s}(\mathbf{W}_{k,u}^{\s}) - \nabla f_\mathrm{s} (\mathbf{W}_{k,u}^{\s}) 
    + \nabla f_\mathrm{s} (\mathbf{W}_{k,u}^{\s}) \right. \right. \nonumber \\
    & \left. \left. - \sigma^2 \mathbf{W}_{k,u}^{\s} - \nabla f(\mathbf{W}_{k,u}^{\s}) \right\|_{\mathrm{F}}^2 \right] \nonumber\\
    \overset{\text{(a)}}{=} & \mathbb{E} \left[\left\| g_{k,u}^{\s}(\mathbf{W}_{k,u}^{\s}) - \nabla f_\mathrm{s} (\mathbf{W}_{k,u}^{\s}) \right\|_{\mathrm{F}}^2 \right] \nonumber \\
    &+ \mathbb{E} \left[\left\| -\sigma^2 \mathbf{W}_{k,u}^{\s} + \nabla f_\mathrm{s} (\mathbf{W}_{k,u}^{\s}) - \nabla f(\mathbf{W}_{k,u}^{\s}) \right\|_{\mathrm{F}}^2 \right], \label{eq:32}
\end{align}
where (a) holds since the error caused by two parts is independent. The RHS of \eqref{eq:32} can be further upper bounded as follows:

\begin{align}   
    & \mathbb{E} \left[\left\| g_{k,u}^{\s}(\mathbf{W}_{k,u}^{\s}) - \nabla f_\mathrm{s} (\mathbf{W}_{k,u}^{\s}) \right\|_{\mathrm{F}}^2 \right] \label{eq:help-2} \\
    \overset{\text{(b)}}{\leq} & \left\| \mathbf{\tilde{X}}^{\TT} \right\|_{\mathrm{F}}^2 \mathbb{E} \left[ \frac{1}{c^2} \left\| \frac{c}{b_{\s}} \mathbf{S}_{k,u}^{\s} - \mathbf{I} \right\|_{\mathrm{F}}^2 \right] \left\| \mathbf{\tilde{X}} \mathbf{W}_{k,u}^{\s} - \mathbf{\tilde{Y}} \right\|_{\mathrm{F}}^2 \nonumber \\
    \overset{\text{(c)}}{\leq} & \frac{c-b_{\s}}{c b_{\s}} \left\| \mathbf{\tilde{X}}^{\TT} \right\|_{\mathrm{F}}^2 \left\| \mathbf{\tilde{X}} \mathbf{W}_{k,u}^{\s} - \mathbf{\tilde{Y}} \right\|_{\mathrm{F}}^2 \nonumber\\
    \leq & \frac{c-b_{\s}}{c b_{\s}} \zeta\kappa, \nonumber
\end{align}
and
\begin{align}
    & \mathbb{E} \left[\left\|  \nabla f_\mathrm{s} (\mathbf{W}_{k,u}^{\s}) - \sigma^2 \mathbf{W}_{k,u}^{\s} - \nabla f(\mathbf{W}_{k,u}^{\s}) \right\|_{\mathrm{F}}^2 \right] \label{eq:help-3} \\
    \overset{\text{(d)}}{\leq} & 4 \left\| \mathbf{X}^{\TT} \right\|_{\mathrm{F}}^2 \mathbb{E} \left[ \left\| \left(\frac{1}{c}\mathbf{G}^{\TT}\mathbf{G}-\mathbf{I}\right) \right\|_{\mathrm{F}}^2 \right] \left\| \mathbf{X} \mathbf{W}_{k,u}^{\s} - \mathbf{Y} \right\|_{\mathrm{F}}^2 \nonumber\\
    &+ 4 \mathbb{E} \left[ \left\| \left(\frac{1}{c}\mathbf{N}^{\TT}\mathbf{N}- \sigma^2 \mathbf{I}\right) \right\|_{\mathrm{F}}^2 \right] \left\| \mathbf{W}_{k,u}^{\s} \right\|_{\mathrm{F}}^2 \nonumber \\
    & + \frac{4\sigma^2}{c^2} \left\| \mathbf{X}^{\TT} \right\|_{\mathrm{F}}^2 \mathbb{E} \left[ \left\| \mathbf{G}^{\TT} \mathbf{N} \right\|_{\mathrm{F}}^2 \right] \left\| \mathbf{W}^{(r)\TT} \right\|_{\mathrm{F}}^2 \nonumber \\
    & + \frac{4\sigma^2}{c^2} \mathbb{E} \left[ \left\| \mathbf{N}^{\TT} \mathbf{G} \right\|_{\mathrm{F}}^2 \right] \big\| \mathbf{X} \mathbf{W}_{k,u}^{\s}  -  \mathbf{Y} \big\|_{\mathrm{F}}^2 \nonumber\\
    \overset{\text{(e)}}{\leq} & \frac{4}{c} (m+m^2) \zeta\kappa + \frac{4}{c} (d+d^2) \phi^2 \sum_{i=1}^{N} \sigma_i^4 \nonumber \\
    &+ \frac{4\sigma^2}{c^2} dmn \zeta\phi^2 + \frac{4\sigma^2}{c^2} dmn \kappa \nonumber \\
    =& \frac{4}{c} [ (m+m^2) \zeta\kappa + (d+d^2) \phi^2 \sum_{i=1}^{N} \sigma_i^4 ] + \frac{4 dmn\sigma^2}{c^2} (\zeta\phi^2 + \kappa), \nonumber
\end{align}
where (b) and (d) follow the inequality $\|\mathbf{A}\mathbf{B}\mathbf{x}\|_{\mathrm{F}}^2 \leq \|\mathbf{A}\|_{\mathrm{F}}^2\|\mathbf{B}\|_{\mathrm{F}}^2\|\mathbf{x}\|_2^2$ for any compatible matrices $\mathbf{A},\mathbf{B}$ and vector $\mathbf{x}$. (c) and (e) hold due to Lemma \ref{useful-lemma} and the fact that the expectation of both $\left\| \mathbf{N}^{\TT} \mathbf{G} \right\|_{\mathrm{F}}^2$ and $\left\| \mathbf{G}^{\TT} \mathbf{N} \right\|_{\mathrm{F}}^2$ equals $dmn$.
By summing up \eqref{eq:help-2} and \eqref{eq:help-3}, we complete the proof.

\section{Proof of Theorem \ref{thm:optimal-r}} \label{proof:contract}

We first present the necessary and sufficient conditions to ensure the feasibility of Problem (P1), which can be shown by following similar lines of the proofs of Lemmas 1-3 in \cite{gao2011spectrum}.

\begin{lemma}\label{thm:sufficient-contract}
A feasible contract $\Omega(\chi)$ for Problem (P1) should satisfy: 1) $\epsilon_1 \geq \epsilon_2 \geq \dots \geq \epsilon_N$; 2) $r_1 \geq r_2 \geq \dots \geq r_N$ and $r_N \geq \mu_{N} \epsilon_{N}> 0$; and 3) $r_{i+1} - \mu_{i+1} \epsilon_{i+1} + \mu_{i+1} \epsilon_{i} \geq r_{i} \geq r_{i+1} - \mu_{i} \epsilon_{i+1} + \mu_{i} \epsilon_{i}$.
\end{lemma}

Then Theorem \ref{thm:optimal-r} can be proved by contradiction.
Assume there exist $\{\tilde{r}_1,\tilde{r}_2,\dots,\tilde{r}_{N} \}$ with at least one $i \in\{1,2,\dots,N\}$ satisfying $\tilde{r}_j < r_i^*$ that achieves a larger utility at the server.
Since $r_i^*= r_{i-1}^* - \mu_{i} \epsilon_{i-1} + \mu_{i} \epsilon_{i}$ and according to the IC constraint, i.e., ${r}_{i} \geq {r}_{i-1} - \mu_{i} \epsilon_{i-1} + \mu_{i} \epsilon_{i}$, we have $\tilde{r}_{i-1} \leq \tilde{r}_{i} + \mu_{i} \epsilon_{i-1} - \mu_{i} \epsilon_{i} = \tilde{r}_{i} - r_i^* + r_{i-1}^* < r_{i-1}^{*}$.
We iterate this inequality to show $\tilde{r}_{N} \!\leq\! r^*_{N} \!=\! \mu_{N}\epsilon_{N}$, which violates the IR constraint of edge device $i$.
Hence, Theorem \ref{thm:optimal-r} is verified.

\section{}\label{sec:table}
Table \ref{table:lambda} shows the relationship between the total reward and values of $\lambda$ and $\sigma^2$.

\begin{table}[!h]
\caption{Relationship between the total reward and values of $\lambda$ and $\sigma^2$.}
\label{table:lambda}
\centering
\resizebox{\columnwidth}{!}{
\begin{tabular}{|c|c|c|c|c|c|c|}
\hline
Total Reward & 70    & 80   & 90 \\ \hline
$\lambda$ (Contract-based)    & 1.8e5 & 8e4  & 5e4  \\
$\sigma^2$ (Contract-based)  & 5,878  & 3,937 & 3,119  \\
$\lambda$ (Stackelberg Game)    & 2e-4 & 1e-5  & 1e-7  \\ 
$\sigma^2$ (Stackelberg Game) & 12,324 & 10,488 & 8,364  \\ \hline
Total Reward & 100  & 110   & 120 \\ \hline
$\lambda$ (Contract-based)   & 9e3  & 3.9e3 & 1e3 \\ 
$\sigma^2$ (Contract-based)  & 1,329 & 876   & 444 \\ 
$\lambda$ (Stackelberg Game)   & 1e-9  & 1.1e-11 & 1e-12 \\ 
$\sigma^2$ (Stackelberg Game) & 6,808 & 5,652 & 5,146 \\ \hline
\end{tabular}}
\end{table}


\end{document}